\documentclass[a4paper,11pt]{article}
\usepackage[left=1.8cm,top=1.8cm, bottom=1.8cm, right=1.8cm]{geometry}
\usepackage{graphicx,subfigure}
\usepackage{soul,color,tikz}

\usepackage[T1]{fontenc}
\usepackage[latin1]{inputenc}
\usepackage{textcomp}
\usepackage{slashed}
\usepackage{amsthm}
\usepackage{amsmath}
\usepackage{amssymb}
\usepackage{graphicx}

\usepackage{authblk}

\usepackage{esint}
\usepackage{tikz}
\usetikzlibrary{positioning}
\usetikzlibrary{calc}
\usetikzlibrary{shapes.geometric}
\usetikzlibrary{decorations.markings}
\usetikzlibrary{arrows,decorations.pathmorphing,backgrounds,positioning,fit,petri,calc,external}
\makeatletter

\usepackage{graphicx,subfigure}
\usepackage{soul,color,tikz}
\usepackage{textcomp}
\usepackage{slashed, url, hyperref}
\usepackage[font=scriptsize]{caption}

 \sethlcolor{yellow}
\usepackage{amsmath}
\usepackage{amssymb}
\usepackage{amsfonts}
\usepackage{amsopn,amsthm}
\usepackage{bbm}
\graphicspath{{img/}}
\newtheorem{theorem}{\bf Theorem}
\newtheorem{lemma}{\bf Lemma}

\newtheorem{conjecture}{Conjecture}
\newtheorem{corollary}{\bf Corollary}
\newtheorem{definition}{\bf Definition}
\newtheorem{example}{Example}

\newtheorem{remark}{\bf Remark}

\def\be{\begin{equation}}
\def\ee{\end{equation}}
\def\bes{\begin{equation*}}
\def\ees{\end{equation*}}
\def\beq{\begin{eqnarray}}
\def\eeq{\end{eqnarray}}
\def\beqs{\begin{eqnarray*}}
\def\eeqs{\end{eqnarray*}}

\newcommand{\enc}{\text{\rm{enc}} }
\newcommand{\dec}{\text{\rm{dec}} }

\def\mx{{\mathcal X}}

 \def\clap#1{\hbox to 0pt{\hss#1\hss}}

\allowdisplaybreaks

\allowdisplaybreaks

\newcommand{\BSC}{\text{\rm{BSC}}}
\newcommand{\BEC}{\text{\rm{BEC}}}
\newcommand{\Bern}{\text{\rm{Bern}}}

\newcommand{\E}{\mathbb{E}}

\begin{document}
\title{Simulation of a Channel with Another Channel}
\author[1]{Farzin Haddadpour}
\author[1]{Mohammad Hossein Yassaee}
\author[2]{Salman Beigi}
\author[1,2]{Amin Gohari}
\author[1]{Mohammad Reza Aref}
\affil[1]{Department of Electrical Engineering, Sharif University of Technology, Tehran, Iran}
\affil[2]{School of Mathematics, Institute for Research in Fundamental Sciences (IPM), Tehran, Iran}
\date{}
\maketitle
\begin{abstract}
In this {paper}, we study the problem of simulating a discrete memoryless channel (DMC)  from another DMC under an average-case and an exact model. We present several achievability and infeasibility results, with tight characterizations in special cases. In particular for the exact model, we fully characterize when a binary symmetric channel (BSC) can be simulated from a binary erasure  channel (BEC)  when there is no shared randomness. We also provide infeasibility and achievability results for simulation of a  binary channel  from another binary channel in the case of no shared randomness.  To do this, we use properties of R\'enyi capacity of a given order. We also introduce a notion of ``channel diameter" which is shown to be additive and satisfy a data processing inequality. 
\end{abstract}

\begin{keywords}
Channel Simulation, Coordination, Point to Point channel, Broadcast, BIBO, OSRB.
\end{keywords}

\section{Introduction}\label{sec:Introduction}

Characterizing when a stochastic resource, such as a channel, can simulate another stochastic resource is of theoretical and practical interest. In particular, this general problem relates to the question of how much randomness one can one distill from an imperfect stochastic resource, or how much randomness is required to synthesis a given stochastic resource (e.g. see \cite{channel-sim1, channel-sim2}). In this work, we are primarily interested in channels as stochastic resources, and whether one can use  a channel to simulate another channel. As another example, assume that we have designed an error-correction code for an intended channel. But, it turns out that the actual channel differs from the intended channel. Then, one may wish to augment the wrong channel so that the same error-correction code can be utilized for even the wrong channel.
\color{black}

As a concrete example, consider a scenario in which memoryless copies of a BEC with erasure probability $\epsilon$ from Alice to Bob is available. Alice and Bob aim to use this resource to simulate memoryless copies of a BSC channel with crossover probability $p$. We require that the number of consumed BECs to be equal to the number of generated BSCs. This is clearly possible when $p\in[\frac{\epsilon}{2}, \frac{1}{2}]$, since Bob can degrade the output of the BEC channel by mapping the erasure symbol to $0$ or $1$ with equal probabilities (see Fig.~\ref{fig2news4}). On the other hand, channel simulation is impossible when $p\in \big[0, h^{-1}(\epsilon)\big)$, since the capacity of the consumed erasure channel should be greater than or equal to the capacity of the simulated BSC channel. Then, the question is whether channel simulation is possible when $p\in \big[h^{-1}(\epsilon), \frac{\epsilon}{2}\big)$. Moreover, what would be the answer, if Alice and Bob are additionally provided with shared randomness at a limited rate? The answer to this question depends whether we want to \emph{approximately or exactly} simulate a given channel (the notion of approximate simulation is made precise later). For instance, in the presence of infinite shared random randomness, any two channels of equal capacity can approximately simulate one another \cite{Benet}\cite{CuffISIT}. As a result, in the presence of infinite randomness, approximate BEC to BSC simulation  is possible if and only if $p\in \big[h^{-1}(\epsilon), \frac{1}{2}\big]$ . On the other hand, we show in this work that exact simulation is possible if and only if $p\in \big[\frac{\epsilon}{2}, \frac{1}{2}\big]$ in presence of no shared randomness. Thus, the answers to the above questions depend on the amount of shared randomness and the notion of channel simulation we are considering. 
\color{black}

\begin{figure}
\begin{center}
\begin{tikzpicture}[scale=1.2, thick]
\node at (-2.3,1) (jointdist){$X$};
\node at (0,2)  (commonrdn1){};
\node at (0,-0.2)  (commonrdn2){};
\node  at ( 0,1)  (ch) {};
\node at (3.6,1) (output2){$Y$};
\node at (2.1,1) (output1){$E$};
\node at ( 2,1) (dec) {};
\node at (-2,2) (enc1) {$0$};
\node at (-2,-0.2) (enc2) {$1$};
\node at (2,2) (dec1) {};
\node at (2,-0.2) (dec2) {};
\draw [->] (dec1)--node[above] {$1$}  (3.3,2);
\draw [->] (dec2)--node[above] {$1$}  (3.3,-0.2);
\draw [->]  (commonrdn1) -- node[above] {$1-\epsilon$}  (2,2) ;
\draw [->]  (enc1) -- node[above] {$1$}(commonrdn1);
\draw [->]  (commonrdn2) -- node[below] {$1-\epsilon$}  (2,-0.2) ;
\draw [->] (enc2) --node[below] {$1$} (commonrdn2);
\draw [->] (commonrdn1)--(dec);
\draw [->] (commonrdn2)--(dec);
\draw [->] (output1)--node[below]{$\frac{1}{2}$}(3.3,2);
\draw [->] (output1)--node[above]{$\frac{1}{2}$}(3.3,-0.2);
\end{tikzpicture}
\end{center}
\caption{Simulation of a $\BSC(\epsilon/2)$ from a $\BEC(\epsilon)$ channel.}
\label{fig2news4}
\end{figure}
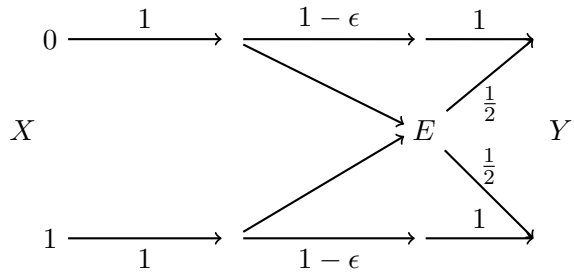

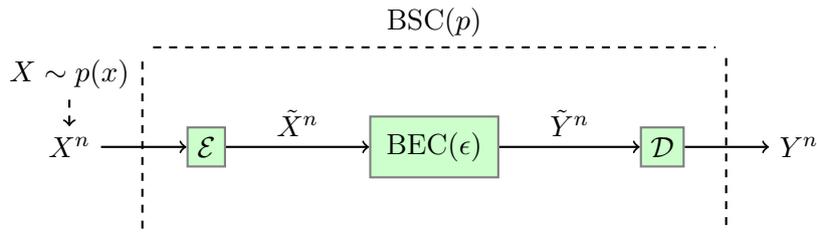
\begin{figure}
\begin{center}
\begin{tikzpicture}[scale=1.2,thick]
\node at (-4,1.8) (dist){$X\sim p(x)$};
\node  at ( 0,1) [rectangle,draw=black!50,fill=green!20!white,inner sep=6pt, minimum size=6mm] (ch) {$\BEC(\epsilon)$};
\node at (-4,1) (input){$X^n$};
\node at (4,1) (output){$Y^n$};
\node at (-3.2,2.1) (nd1){};
\node at (3.2,2.1) (nd2){};
\node at (-3.2,-0.02) (nd3){};
\node at (3.2,-0.02) (nd4){};
\node at ( 2.5,1) [rectangle,draw=black!50,fill=green!20!white] (dec) {$\mathcal D$};
\node at (-2.5,1) [rectangle,draw=black!50,fill=green!20!white](enc) {$\mathcal E$};
\draw [-,dashed](nd1)--node[above]{$\BSC(p)$} (nd2);
\draw [-,dashed](nd2)--(nd4);
\draw [-,dashed](nd4)--(nd3);
\draw [-,dashed](nd3)--(nd1);
\draw [->] (enc) -- node[above] {$\tilde{X}^n$} (ch);
\draw [->] (ch)  -- node[above] {$\tilde{Y}^n$} (dec);
\draw [->] (input)--(enc);
\draw [->] (dec)--(output);
\draw [->,dashed] (dist)--(input);
\end{tikzpicture}
\caption{Channel simulation in absence of common randomness.}
\end{center}
\vskip -0.13cm
\label{fig:cs-no-random}
\end{figure}

The channel simulation problem can be described as follows: Alice has an $x^n$ sequence, and $n$ memoryless copies of
the channel $p(\tilde{y}|\tilde{x})$ as a resource. She creates $\tilde{x}^{{n}}$ as a (stochastic) function of $x^n$, and sends it to Bob over $n$ copies of the channel $p(\tilde{y}|\tilde{x})$. Bob receives $\tilde{y}^{{n}}$ and passes it through another stochastic map to generate $y^n$. Their goal is that the induced channel, which maps $x^n$ to $y^n$,  to be exactly equal, or $\epsilon$-close to $n$ copies of a target memoryless channel $p(y|x)$, for some sufficiently small $\epsilon$. 

To formally define the simulation problem, \color{black}one has to specify a model for the input sequences, $x^n$. Furthermore, one should also specify the simulation error $\epsilon$. Here, we consider two models of \emph{worst-case} and \emph{average-case} for the input sequence, and two models of \emph{asymptotically reliable} ($\epsilon\rightarrow 0$) and \emph{exact} ($\epsilon=0$) for the error. Briefly speaking, 
in the {worst-case model} we demand reliable simulation \emph{for all} possible input sequences $x^n$. In the {average-case model}, on the other hand, we assume that the input $X^n$ is i.i.d.\ according to some distribution $p(x)$.  We use the total variation distance between the simulated channel and the desired channel to measure the performance of the simulation protocol. In an asymptotically reliable simulation, we want the total variation distance to vanish asymptotically, whereas in the exact simulation we want it to be exactly zero. All in all, three models of channel simulation emerges: 
\begin{itemize}
\item average-case  with asymptotically vanishing distortion,
\item worst-case  with  asymptotically vanishing distortion,
\item exact model. 
\end{itemize}
In the exact model, the zero distortion constraint implies exact simulation for all possible input sequences~$x^n$. The average-case model would require approximate channel simulation for most typical sequences~$x^n$. Finally, the worst-case model with asymptotically vanishing distortion requires approximate channel simulation for all sequences~$x^n$. 


Another aspect of the channel simulation problem is the existence or lack thereof of shared randomness. Thus, any of the models can be conceived when limited share randomness exists between the transmitter and the receiver.
We would like to emphasize that one can also conceive other notions of channel simulation, which we do not consider in this work, \emph{e.g.,} average distortion ($\bar{d}$-distance) measure as in~\cite{Neuhoff1,Neuhoff2,channel-sim1}, agreement probability as in \cite{andrej}, \color{black} the  {empirical coordination} of~\cite{coordinationcuff} or the exact simulation of \cite{Kumar}; see also~\cite{exactcoordination, YassaaeeInteractive}.

While the problem of simulating an arbitrary channel from another one has been formally defined in the literature, \emph{e.g.} in \cite{Winter1}, but only the following special cases are studied in the literature:

\vspace{.12in}
\noindent\textbf{Simulation of a noiseless link:} Suppose that we would like to simulate a noiseless link from a given channel $p(\tilde y| \tilde x)$. In both the average-case and worst-case asymptotic models, a noiseless link can be simulated with $p(\tilde y| \tilde x)$ only if its rate is less than the Shannon capacity of $p(\tilde y| \tilde x)$. Indeed, channel simulation in this special case, corresponds to 
the problems of computing the average probability of error capacity, and maximum probability of error capacity of $p(\tilde y| \tilde x)$ which are known to be equal. 
Similarly, for exact simulation, the rate of the noiseless link should be less than the zero-error capacity of $p(\tilde y| \tilde x)$.  Shared randomness does not affect the zero-error and average error capacity of a channel, hence is useless for simulation of a noiseless link. 

\vspace{.12in}
\noindent\textbf{Simulation with a noiseless link:} Conversely, suppose that we would like to simulate a channel $p(y|x)$ from a noiseless link as well as limited shared randomness. The reverse Shannon theorem shows the ability of a noiseless link to simulate a noisy channel of equal
capacity  in the presence of infinite shared randomness for the average-case asymptotic model \cite{Benet}. 
In other words, a noiseless link whose rate is equal to the capacity of the channel $p(y|x)$, is sufficient for simulating the channel $p(y|x)$. Therefore, all channels with the same capacity are equivalent in the presence of infinite shared randomness for the average-case asymptotic model.  However, the assumption of infinite shared randomness is crucial here. 
This is demonstrated by the result of Cuff in \cite{CuffISIT}. \\ \noindent
In the exact model, we saw above that simulation of a noiseless link from a given channel reduces to the zero-error capacity of a channel, and the amount of shared randomness is not relevant. Unfortunately computing the zero-error capacity of a channel is an open problem. Fortunately, the reverse problem has a clean solution when infinite shared randomness is available, given in \cite{Winter1}. While the  authors  in \cite{Winter1} do not explicitly express their solution in terms of the  R\'enyi mutual information, one can observe that their answer is the 
 R\'enyi capacity of order infinity of the channel (see \cite{exactcoordination} for a discussion).
\color{black}

\begin{figure}
\begin{center}
\begin{tikzpicture}[scale=1.4, thick]
\node at (-3,2) (dist){};
\node at (0,1.8)  (commonrdn){$\omega$};
\node  at ( 0,1) [rectangle,draw=black!50,fill=green!20!white,inner sep=6pt, minimum size=6mm] (ch) {$p(\tilde{y}|\tilde{x})$};
\node at (-3,1) (input){$x^n$};
\node at (3,1) (output){$y^n$};
\node at ( 2,1) [rectangle,draw=black!50,fill=green!20!white] (dec) {$\mathcal{D}$};
\node at (-2,1) [rectangle,draw=black!50,fill=green!20!white](enc) {$\mathcal{E}$};
\draw [->] (enc) -- node[above] {$\tilde{x}^{n}$} (ch);
\draw [->] (ch)  -- node[above] {$\tilde{y}^{n}$} (dec);
\draw [->,dashed]  (commonrdn) to  (2,1.8) to (dec);
\draw [->,dashed]  (commonrdn) to  (-2,1.8) to (enc);
\draw [->] (input)--(enc);
\draw [->] (dec)--(output);
\end{tikzpicture}
\label{fig1general}
\end{center}

\begin{center}
\begin{tikzpicture}[scale=1.4, thick]
\node at (-3,1.8) (dist){$X\sim p(x)$};
\node at (0,1.8)  (commonrdn){$\omega$};
\node  at ( 0,1) [rectangle,draw=black!50,fill=green!20!white,inner sep=6pt, minimum size=6mm] (ch) {$p(\tilde{y}|\tilde{x})$};
\node at (-3,1) (input){$X^n$};
\node at (3,1) (output){$Y^n$};
\node at ( 2,1) [rectangle,draw=black!50,fill=green!20!white] (dec) {$\mathcal{D}$};
\node at (-2,1) [rectangle,draw=black!50,fill=green!20!white](enc) {$\mathcal{E}$};
\draw [->] (enc) -- node[above] {$\tilde{X}^{n}$} (ch);
\draw [->] (ch)  -- node[above] {$\tilde{Y}^{n}$} (dec);
\draw [->,dashed]  (commonrdn) to  (2,1.8) to (dec);
\draw [->,dashed]  (commonrdn) to  (-2,1.8) to (enc);
\draw [->] (input)--(enc);
\draw [->] (dec)--(output);
\draw [->,dashed] (dist)--(input);
\end{tikzpicture}
\caption{Simulation of $p(y|x)$ using $p(\tilde{y}|\tilde{x})$ using shared randomness. (Top figure) the general channel simulation problem for a point-to-point channel (Bottom figure) In the average-case model, the input to the channel, i.e., $X$ is assumed to be i.i.d.}
\label{fig1}
\end{center}
\end{figure}
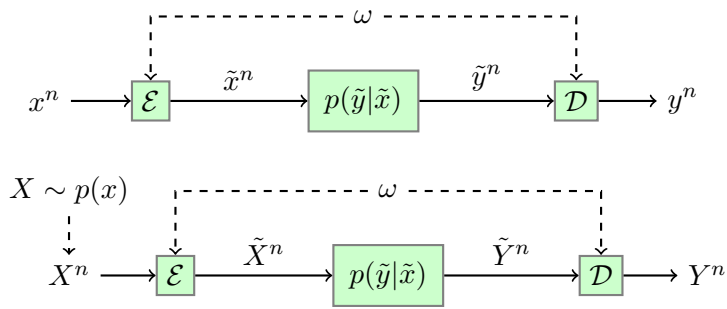

\subsection{Main results}
In this paper, we only study the two models of exact and asymptotic average-case. We also study the problem of simulating a broadcast channel with another broadcast channel. 
The main contributions of this paper are as follows.

\vspace{.12in}
\noindent\textbf{Asymptotic average-case model:} 
The channel simulation problem in the asymptotic average-case model is essentially a joint source-channel coding problem as we compress and code $X^n$ into the channel inputs $\tilde{X}^n$. From a mathematical perspective, the problem is complicated due to the following fact: consider the joint pmf on random variables $X^n, \tilde{X}^{{n}}, \tilde{Y}^{{n}}, Y^n$ induced by a simulation code of length $n$. In the case of no-shared randomness, we have the Markov chain $X^n\rightarrow \tilde{X}^{{n}}\rightarrow \tilde{Y}^{{n}}\rightarrow Y^n$. However, $X^n, \tilde{X}^{{n}}, \tilde{Y}^{{n}}, Y^n$ do not necessarily have a joint i.i.d. distribution (even in the average-case model we only require $X^n, Y^n$ to be almost i.i.d.).

For this model, we provide an inner bound for the point-to-point channel simulation problem using the OSRB technique~\cite{OSRB}. This inner bound is based on a ``hybrid coding" scheme to perform joint source-channel coding. We also provide an outer bound. Then, we compare these bounds for simulation of a $\BSC$ channel from a $\BEC$ channel in the absence of shared common randomness. We show that in this example, the degradation strategy is sub-optimal for any non-uniform input distribution.

\vspace{.12in}
\noindent\textbf{Exact model:}  Our main result here are two infeasibility results in the absence of shared common randomness. Our first result is based on the observation that if channel $p(y|x)$ can be exactly simulated using $p(\tilde{y}|\tilde{x})$, then the channel $p(\tilde{y}|\tilde{x})$ must have a better \emph{error exponent} than $p(y|x)$ for all communication rates. Error exponents are known to have a characterization in terms of  R\'enyi capacity. Our first infeasibility result is also in terms of  R\'enyi capacity. Our second infeasibility result is in terms a quantity that we introduce and name the ``channel diameter".  Just like the  R\'enyi capacity of a channel, this quantity is  shown to be additive and satisfy a data processing inequality. We use these two infeasibility results to 
\color{black}
show that the  degradation strategy is optimal for simulation of a $\BSC$ channel from a $\BEC$ channel in the absence of shared common randomness. We also study when we can simulate a binary input binary output channel (BIBO) from another BIBO channel.

\vspace{.12in}
\noindent\textbf{Broadcast channel simulation:}  We consider an extension of the point-to-point channel simulation problem under the asymptotic average-case model. We apply the OSRB technique to the broadcast channel simulation problem and present an inner bound for it. Simulation of a broadcast network was stated as an open problem in~\cite[p.100]{CuffThesis}.

\vspace{.12in}
The rest of the paper is organized as follows: In Section \ref{Notation} we set up the notation. In Section \ref{two-models}, two models for channel simulation are introduced. In Sections \ref{p-p-i} and \ref{p-p-o} we provide our results for a point-to-point channel. Section \ref{sec:StrongBroadcast} contains our results for the broadcast channel. In Section \ref{Zsec} our result for exact channel simulation is presented. The proofs are included in Section \ref{sec:proofs}.


\section{Notation and preliminaries}\label{Notation}
Throughout this paper, we restrict ourselves to discrete random variables, taking values in finite sets. We generally use $q(\cdot)$ to denote the pmf's induced by a code, and $p(\cdot)$ to denote the desired i.i.d.~pmf's. We also use $\omega$ as a random variable to denote shared randomness. To distinguish the uniform distribution we use the notation $p^{\mathsf u}$. The total variation distance between two pmf's $p$ and $q$ on the same alphabet $\mathcal{X}$, is denoted by $\|p\text{\textminus}q\|_{1}$. We say that $p\overset{\epsilon}{\approx}q$ if $\|p\text{\textminus}q\|_{1}\leq \epsilon$.  A sequence of random variables $X_1, X_2, \cdots, X_j$ is denoted by $X_{[1:j]}$ where $[1:j]=\{1,2,\cdots, j\}$, and $\bar{\alpha}=1-\alpha$. We use $H(X)$ to denote the entropy of a random variable $X$ and $h(p)$ to denote the binary entropy of a Bernoulli random variable with parameter~$p$. All the logarithms in this paper are in base~2.

We frequently use the concept of random pmf's, which we denote by capital
letters (e.g., $P_X$). To be more precise, let $\Delta^{\mathcal{X}}$ be the probability simplex over the set $\mathcal{X}$. A random pmf
$P_X$ is indeed a probability distribution over $\Delta^{\mathcal{X}}$. In other words, if we use $\Omega$ to denote some sample space, we may have a random variable $P_{X}(x;\omega)$ such that 
for any $\omega \in \Omega$ and $x\in \mathcal{X}$ we have $P_{X}(x;\omega)\geq 0$, and $\sum_{x}P_{X}(x;\omega)=1$ for all $\omega$. Then, $P_{X}(.;\omega)$ is a vector of random variables which resembles a random distribution. We denote such a random pmf by $P_X$. 

We can definite random $P_{X,Y}$ on a product set $\mathcal{X}\times \mathcal{Y}$ similarly. We note that we can continue to use the law of total probability with random pmf's
(e.g., to write $P_{X}(x)=\sum_{y}P_{X,Y}(x,y)$ meaning that $P_{X}(x;\omega)=\sum_{y}P_{X,Y}(x,y;\omega)$ for all $\omega$) and conditional probability pmf's (e.g., to write $P_{Y|X}(y|x)=\frac{P_{XY}(x,y)}{P_{X}(x)}$ meaning that $P_{Y|X;\omega}(y|x)=\frac{P_{XY}(x,y;\omega)}{P_{X}(x;\omega)}$ for all $\omega$).

\subsection{R\'enyi divergence}

\begin{definition} R\'enyi divergence of order $\alpha$, or the $\alpha$-R\'enyi divergence, for $\alpha\in (0,1)\cup (1,\infty)$, between two pmf's is defined as
$$D_\alpha(p\| q) := \frac{1}{\alpha -1} \log \big(\sum_x p_x^{\alpha}q_x^{1-\alpha}\big).$$ 
We also define
\begin{align}
H_{\alpha}(p) := - D_{\alpha}(p\| I) = \frac{1}{1-\alpha}\log\big(\sum_x p_x^{\alpha}\big).\nonumber
\end{align}
\end{definition}
Note that $\underset{\alpha\rightarrow1}{\lim}D_{\alpha}(p\|q)=D(p\|q)$ as well as $\underset{\alpha\rightarrow1}{\lim}{H_{\alpha}(p)}=H(p)$ which are KL divergence and Shannon entropy, respectively. 

It is known that R\'enyi divergence satisfies the data processing inequality similar to KL divergence (see e.g.,~\cite{Ervin}). 

There are several definitions for  mutual information  of order $\alpha$. Here we use Sibson's choice~\cite{Sibson} as follows: 
\begin{align*}
I_{\alpha}(X;Y)&:=\min_{q_Y} D_\alpha( p_{XY}\| p_X\cdot q_Y)\\
&=\frac{\alpha}{\alpha-1} \log\left(\sum_y \left[\sum_x p(x)p(y|x)^{\alpha}\right]^{1/\alpha}\right).
\end{align*}
The $\alpha$-capacity of a channel $p(y|x)$ is then defined as (see~\cite{Verdu} for a review):
\begin{align}
C_{\alpha}\big(p(y|x)\big)&=\max_{p(x)} I_\alpha(X; Y)\nonumber\\
&=\max_{p(x)}\frac{\alpha}{\alpha-1} \log\left(\sum_y \left[\sum_x p(x)p(y|x)^{\alpha}\right]^{1/\alpha}\right).
\label{eq:alpha-capacity}
\end{align}
In particular, the capacity of order infinity ($\infty$-capacity) is equal to
\begin{align}
C_{\infty}(p(y|x))=\log\left(\sum_{y}\underset{x}{\max}\;p(y|x)\right)\label{RenyMut}.
\end{align}
For $\alpha=1$, we let $C_1(p(y|x))=C(p(y|x))$ to be the Shannon's capacity of the channel.

Interestingly, just like the Shannon capacity, the $\alpha$-capacity is also additive for product channels.
\begin{theorem}[\cite{Arimoto}]\label{thm:prodal} For $\alpha>0$ and product of identical channels $p(y^n|x^n)=\prod_{i=1}^np(y_i|x_i)$ we have 
$$C_{\alpha}\big(p(y^n|x^n)\big)=n C_{\alpha}\big(p(y|x)\big).$$
\end{theorem}
Using the data processing property for R\'enyi divergence, one can easily show the data processing property for mutual information of order $\alpha$:
\begin{theorem}[\cite{Polyanskiy}] If $X-Y-Z-W$ and $\alpha>0$ we have
$$I_\alpha(Y;Z)\leq I_\alpha(X;W).$$
\label{thm:dataal}
\end{theorem}
\begin{lemma}\label{lemma:rev-added}
$C_{\alpha}\big(p(y|x)\big)$ is quasi-convex in $p(y|x)$ for any $\alpha>0$.
\end{lemma}
\begin{proof}
The mapping $p(y|x)\mapsto \zeta_\alpha(I_\alpha(X; Y))$ for any fixed $p(x)$ is  convex, where 
$$\zeta_\alpha=\frac{1}{\alpha-1}\exp(t-\frac{t}{\alpha})$$
is a monotonically increasing function \cite[Theorem 4]{Verdu}. Since the inverse of $\zeta_\alpha$ is also a  monotonically increasing function, and the composition of any convex function with an increasing function is quasi-convex, we conclude that the mapping $p(y|x)\mapsto I_\alpha(X; Y)$  is  quasi-convex for any fixed $p(x)$. This implies that the mapping $p(y|x)\mapsto C_\alpha(p(y|x))$ is a maximum of quasi-convex functions, and hence quasi-convex itself.
\end{proof}

\subsubsection{Optimal error exponent}
 R\'enyi mutual information finds an operational meaning in connection to the optimal error exponent of a channel. Given a channel $p(y|x)$ and a fixed transmission rate $R$, let $P_e^{(n)}$ denote the error probability of the best code with blocklength $n$. Then, the optimum error exponent $E_{p(y|x)}(R)$, at some fixed transmission rate $R$, is the supremum of $-\frac{1}{n}\ln P_e^{(n)}$ as $n$ tends to infinity. The optimal error exponent $E(R)$ is also known as the channel reliability function. The channel reliability function is the same under average and maximal notions of error probability; this can be proved using the standard argument of pruning a code with low average probability of error to construct a code with low maximal probability of error. It is known that $E_{p(y|x)}(R)$ is related to capacity of order $\alpha$:

\begin{theorem}\cite{Gallager1, Gallager2} For any rate $R$ less than the channel capacity $C$, we have that 
$$ \max_{\alpha\in [\frac 12 ,1]}(\alpha^{-1}-1)\left(C_{\alpha}(p(y|x))-R\right) \leq E_{p(y|x)}(R)\leq \max_{\alpha\in (0 ,1]}(\alpha^{-1}-1)\left(C_{\alpha}(p(y|x))-R\right).$$
\end{theorem}
\color{black}


\subsection{Two models for channel simulation}\label{two-models}
Consider the problem of simulating memoryless copies of the channel $p(y|x)$ given memoryless copies of $p(\tilde{y}|\tilde{x})$ as depicted in Fig.~\ref{fig1}. Alice's input is denoted by~$X^n$. The shared randomness between Alice and Bob is denoted by the random variable $\omega$, which is independent of $X^{n}$ and uniformly distributed over $[1:2^{nR}]$. A simulation code consists of
\begin{itemize}
\item A (stochastic) encoder with conditional pmf $q^{\enc}(\tilde{x}^{n}|x^{n},\omega)$.
\item A (stochastic) decoder with conditional pmf $q^{\dec}(y^{n}|\tilde{y}^{n},\omega)$ .
\end{itemize}
Thus, the joint distribution induced by the code is as follows:
\begin{align}\label{eq:induced-dist-2}
q( \tilde{x}^{n}, \tilde{y}^{n}, y^n|\omega,x^n)=q^{\enc}(\tilde{x}^{n}|x^{n},\omega)\Big(\prod_{i=1}^{n}p(\tilde{y}_i|\tilde{x}_i)\Big)q^{\dec}(y^{n}|\tilde{y}^{n},\omega).
\end{align}

 An $(n, R)$ \emph{exact} simulation code requires that for every $x^n$ and $y^n$
\begin{align}
q(y^n|x^n)=\prod_{i=1}^np(y_i|x_i),
\end{align}
where $q(y^n|x^n)$ is the pmf induced by the code.

On the other hand, an \emph{average-case} simulation code assumes that Alice observes i.i.d.\ copies of a source $X$ (taking values in the finite set $\mx$ and having pmf $p(x)$). Then an $(n,R,\epsilon)$ average-case simulation code would require that
$$
\Big\| q(x^{n},y^{n})-\overset{n}{\underset{i=1}{\prod}}p(x_{i})p(y_{i}|x_i)\Big\|_{1}\leq \epsilon.
$$

\begin{definition}
The channel $p(y|x)$  is said to be in the admissible region of the channel-rate pair $\big(p(\tilde{y}|\tilde{x}), R\big)$ for the exact simulation model if one can find an $(n,R)$ exact simulation code for some $n$. \\ \noindent
The input distribution-channel pair $\big(p(x), p(y|x)\big)$ is said to be in the admissible region of the channel-rate pair $\big(p(\tilde{y}|\tilde{x}), R\big)$ for the asymptotic average-case model if one can find a sequence of $(n,R,\epsilon_n)$ average-case simulation codes such that
$
\underset{n\rightarrow \infty}{\lim} \epsilon_n=0
$.
\end{definition}

Clearly, exact channel simulation implies average case simulation for any arbitrary $p(x)$.


\section{Average-case model}

In this section we state our results on the channel simulation problem in the average-case model. For the proofs of these results refer to  Section~\ref{sec:proofs}.

\subsection{Point-to-point channel: inner bound}\label{p-p-i}
\begin{theorem}\label{Pt2PtIn} \emph{(Inner Bound)} A pair $(p(x),p(y|x))$ is in the admissible region of the channel-rate pair $\big(p(\tilde{y}|\tilde{x}), R\big)$ if
one can find channels $p(u, \tilde{x}|x)$ and $p(y|\tilde{y},u)$ such that 
\begin{align}\label{eq:ppi-factor}
p(u,\tilde{x},\tilde{y},x,y)=p(x)p(u,\tilde{x}|x)p(\tilde{y}|\tilde{x})p(y|\tilde{y},u),
\end{align}
with the given marginal $p(x,y)=p(x)p(y|x)$, and that
{{\begin{align}
R+I(U;\tilde{Y})&>I(U;XY),\nonumber\\
I(U;\tilde{Y})&>I(U;X).\label{thm1region}
\end{align}}}
\end{theorem}

The main idea of the above theorem is to use an encoding scheme similar to hybrid codes, by utilizing an auxiliary random variable to allow a coupling between the channels, as opposed to simply reducing the resource channel to a noiseless link according to its channel capacity (see \cite{CuffS}
for another uses of hybrid coding).
\color{black}

Let us examine the following (non-optimal) general strategy for channel simulation and compare it to the above theorem. 
We can always convert the channel $p(\tilde{y}|\tilde{x})$ to a noiseless link of rate $\tilde{C}=\max_{p(\tilde{x})}I(\tilde{X};\tilde{Y})$ and then use the achievability part of the result of \cite{CuffISIT} (described in the introduction) to simulate the channel $p(y|x)$ from the noiseless link. This strategy is a special case of the above theorem. Let $p(\tilde x)$ denote the pmf that maximizes $I(\tilde{X};\tilde{Y})$. Also let $U$ in the above theorem to be of the form $U=(\tilde{X}, U')$ where $(U', X, Y)$ is independent of $\tilde{X}$ and satisfies $X-U'-Y$.  Then~\eqref{eq:ppi-factor} holds and~\eqref{thm1region} reduces to
\begin{align}
R+\tilde{C}&>I(U';XY),\nonumber\\
\tilde{C}&>I(U';X).
\label{eq:cuff-bound-4}
\end{align}
This region has appeared in~\cite{CuffISIT}.

In the special case of $\tilde{Y}=f(\tilde{X})$ where $\tilde Y$ is a function of $\tilde X$, our inner bound matches~\eqref{eq:cuff-bound-4}, which is expected due to the converse given in~\cite{CuffISIT}. To see this, take an arbitrary $p(u,\tilde{x},\tilde{y},x,y)$ satisfying~\eqref{eq:ppi-factor}. In this case we have $X-U\tilde{Y}-Y$, i.e., we have $X-U'-Y$ for $U'=(U,\tilde{Y})$. As a result,  
\begin{align*}I(U';XY)&=I(U;XY)+I(\tilde{Y};XY|U)\\
&\leq I(U;XY)+H(\tilde{Y}|U)\\&= I(U;XY)-I(U;\tilde{Y})+H(\tilde{Y})\\&<R+I(U;\tilde{Y})-I(U;\tilde{Y})+H(\tilde{Y}).
\end{align*}
Thus, $R+H(\tilde{Y})=R+I(\tilde{X};\tilde{Y})>I(U';XY)$, and hence $R+\tilde{C}>I(U';XY)$. To verify the second inequality, observe that
\begin{align*}I(U';X)&=I(U;X)+I(\tilde{Y};X|U)\\&\leq I(U;X)+H(\tilde{Y}|U)\\&=I(U;X)-I(U;\tilde{Y})+H(\tilde{Y})
\\&<I(U;\tilde{Y})-I(U;\tilde{Y})+H(\tilde{Y})\\&=H(\tilde{Y}).
\end{align*}
This means that $\tilde{C}\geq I(\tilde{X};\tilde{Y})=H(\tilde{Y})>I(U';X)$.

While the cardinality of $U$ in the above theorem can be bounded from above by $|\mathcal{X}||\mathcal{Y}||\mathcal{\tilde{X}}||\mathcal{\tilde{Y}}|$, it is not easy to explicitly compute the inner bound for a given arbitrary pair of channels. In particular, it is not easy to compute the achievable region for the BEC-BSC example that we discussed in the introduction. Nonetheless, we show that the above inner bound gives a strictly better strategy than the degradation scheme, for any non-uniform input pmf in the BEC to BSC simulation problem.\color{black}

\begin{theorem}\label{Subth}
Degradation for simulating a $\BSC(p)$ from a $\BEC(\epsilon)$ in the average-case model with non-uniform input pmf is suboptimal. In other words,  given any $\epsilon\in (0,1)$ and any non-uniform binary input pmf $p(x)$, there exists $p<{\epsilon}/{2}$ such that $(p(x), \BSC(p))$ is in the admissible region of $(\BEC(\epsilon), 0)$.
\end{theorem}


\subsection{Point-to-point channel: outer bound}\label{p-p-o}


\begin{theorem}[Outer Bound] \label{converseThm} If the input distribution-channel pair $(p(x), p(y|x))$ is in the admissible region of the channel-rate pair $\big(p(\tilde{y}|\tilde{x}), R\big)$, then for any non-negative reals $\beta$, $\gamma$, and $\theta$ we have
\begin{align}
I(X;Y)+&\min_{U: X-U-Y}\big[\beta I(U;XY)+\gamma I(U;X)+\theta I(U;Y)\big]\leq \nonumber \\
&\max_{p(\tilde{x})}\bigg[I(\tilde{X};\tilde{Y})+\min_{\tilde{U}: \tilde{X}-\tilde{U}-\tilde{Y}}\big[\beta I(\tilde{U};\tilde{X}\tilde{Y})+\gamma I(\tilde{U};\tilde{X})+\theta I(\tilde{U};\tilde{Y})\big]\bigg]
+(\beta+\gamma+2\theta)R.\color{black}
\label{eqn:tm2}
\end{align}
Further, to compute the above minimums one can put the cardinality bounds of $|\mathcal{U}|\leq |\mathcal{X}|\cdot |\mathcal{Y}|$ and $|\tilde{\mathcal{U}}|\leq |\tilde{\mathcal{X}}|\cdot |\tilde{\mathcal{Y}}|$.
\end{theorem}

According to~\cite{Benet, CuffISIT}, when  infinite shared randomness is available, $\tilde C=C(p(\tilde y|\tilde x))\geq I(X;Y)$ is a sufficient and necessary condition for channel simulation. By setting $\beta=\gamma=\theta=0$ in the above theorem, we recover the necessity of this condition. \color{black}

\begin{corollary} $(p(x), p(y|x))$ is in the admissible region of $\big(p(\tilde{y}|\tilde{x}), \infty\big)$ if and only if 
$\tilde C=C(p(\tilde y|\tilde x))\geq I(X;Y)$.
\end{corollary}

A channel $p(y|x)$ is called symmetric if for every permutation $\pi_X$ on $\mathcal{X}$, there is a permutation $\pi_Y$ on $\mathcal{Y}$ such that $p(\pi_Y(y)|\pi_X(x))=p(y|x)$. For such channels, the outer bound of the above theorem can be simplified as the maximum on the right hand side occurs at uniform distribution $p^{\mathsf u}(\tilde x)$.

\begin{theorem}[Outer Bound for Symmetric Channels] \label{converseThm2}
Suppose that the channel $p(\tilde{y}|\tilde{x})$ is symmetric. If $(p(x), p(y|x))$ is in the admissible region of $\big(p(\tilde{y}|\tilde{x}), 0\big)$, then $I_{p^{\mathsf u}(\tilde x)}(\tilde{X};\tilde{Y})\geq I(X;Y)$  and $\mathcal{S}(p(\tilde y|\tilde x), p^{\mathsf u}(\tilde x))\subseteq \mathcal{S}(p(y|x), p(x))$ where
$$\mathcal{S}(p(y|x), p(x))\triangleq\bigcup_{X-U-Y}\big\{(a_1, a_2, a_3): a_1\geq I(U;XY), a_2\geq I(U;X), a_3\geq I(U;Y)\big\}.$$
\end{theorem}

\begin{remark}
\emph{ The minimum value of $a_1$ such that $(a_1, \infty, \infty)\in\mathcal{S}(p(y|x), p(x))$ is Wyner's common information between $X$ and $Y$, denoted by $C(X,Y)$.} 
\end{remark}

 As an application of the above theorem, we show that for any $0<p<1$ and $\epsilon<1/2$, it is impossible to simulate a $\BEC(\epsilon)$ with uniform input distribution from a $\BSC(p)$ when there is no shared randomness. In other words, we show that $(p^{\mathsf u}, \BEC(\epsilon))$ is not admissible for $(\BSC(p), 0)$. To prove this, observe that
the simple mutual information bound is inadequate as it gives $1-\epsilon<1-h(p)$ which does not exclude the possibility of simulation for the entire range of $0<p<1$. Nevertheless, if simulation is possible, from the above outer bound for symmetric channels, we conclude that the Wyner common information of $\BEC(\epsilon)$ with uniform input distribution must be less than or equal to the maximum of the Wyner common information of $\BSC(p)$ over all input distributions. For $\epsilon<1/2$, the Wyner common information of a $\BEC(\epsilon)$ with uniform input is one~\cite[IV.\ A]{CuffISIT}. On the other hand, for any input distribution $p(x)$ the Wyner common information $C(X,Y)$ is strictly less than one, since when  $p\in(0,1)$, one can write the $\BSC(p)$ channel as the cascade of two \emph{non-trivial} $\BSC$ channels $X\rightarrow U\rightarrow Y$. Then $I(U;XY)$ will be strictly less than $H(U)$ which is less than one. This concludes our claim. Since the average-case model is less restrictive than the exact model, we can conclude that simulating a $\BEC(\epsilon)$ from a $\BSC(p)$ is impossible under the exact model as well when $0<p<1$ and $\epsilon<1/2$. 

In the following theorem we characterize the set $\mathcal{S}(p(y|x), p(x))$ for $\BSC$ and $\BEC$ channels with uniform input distributions.

\begin{theorem} 
\begin{itemize}
\item[\rm{(i)}] $\mathcal{S}(\BEC(\epsilon), p^{\mathsf u})$ is the set of triples $(a_1, a_2, a_2)$ such that for some $a\in[1-\epsilon, 1]$ we have 
\begin{align*} 
a_1&\geq -\epsilon\log(\epsilon)+a-a\log(a)+(a-1+\epsilon)\log({a-1+\epsilon}),
\\
a_2&\geq a,
\\
a_3&\geq 1-\epsilon-\epsilon\log(\epsilon)-a\log(a)+(a-1+\epsilon)\log({a-1+\epsilon}).
\end{align*}

\item[\rm{(ii)}] Assuming {Conjecture \ref{conjecture1}} (stated in Section~\ref{sec:proofs}), the set $\mathcal{S}(\BSC(p), p^{\mathsf u})$  consists of triples $(a_1, a_2, a_3)$ such that for some $0\leq \alpha, \beta\leq 1$ with $ p = \alpha \bar \beta+ \bar \alpha  \beta$ we have 
\begin{align*}
 a_1& \geq 1+h(p)-h(\alpha)-h(\beta),
\\a_2&\geq 1-h(\alpha),
\\a_3&\geq 1-h(\beta).
\end{align*}
\end{itemize}
\label{thm3}
\end{theorem}
\begin{example}
Let us use the above theorem as well as Theorem~\ref{converseThm2} to study the problem of simulating a BSC channel from a BEC channel. More precisely, we want to characterize the range of $p\in[0,1/2]$ such that $(p^{\mathsf u},\BSC(p))$ is in the admissible region of $(\BEC(\epsilon), 0)$. The degradation scheme (of Fig.~\ref{fig2news4}) for this problem gives the achievability of $p=\epsilon/2$.
On the other hand, if simulation is possible, the trivial mutual information bound implies that $1-\epsilon\geq 1-h(p)$, which does not match the achievability bound. Nevertheless, the above results give a stronger converse. 

If simulation is possible, then for any $a\geq 1-\epsilon$ there exist $\alpha$ and $\beta$ such that $ p = \alpha \bar \beta+ \bar \alpha  \beta$ and
\begin{align}
-\epsilon\log(\epsilon)+a-a\log(a)+(a-1+\epsilon)\log({a-1+\epsilon})&\geq 1+h(\alpha\bar{\beta}+\bar{\alpha}\beta)-h(\alpha)-h(\beta),\label{eqn:anewsb1}\\
a&\geq 1-h(\alpha),\label{eqn:anewsb2}\\
1-\epsilon-\epsilon\log(\epsilon)-a\log(a)+(a-1+\epsilon)\log({a-1+\epsilon})&\geq 1-h(\beta).\label{eqn:anewsb3}
\end{align}
In other words, simulation is possible only if $p\geq p^*$ with
$$p^*=\underset{a\geq 1-\epsilon}{\max}\:\underset{\alpha,\beta}{\min} \:{(\alpha\bar{\beta}+\bar{\alpha}\beta)},$$
where the minimum is over $\alpha,\beta \in[0,1]$ satisfying \eqref{eqn:anewsb1}-\eqref{eqn:anewsb3}. Observe that without loss of generality we can restrict to $\alpha,\beta\in [0,1/2]$.

\begin{figure}[t]
\begin{center}
\includegraphics[scale=0.4]{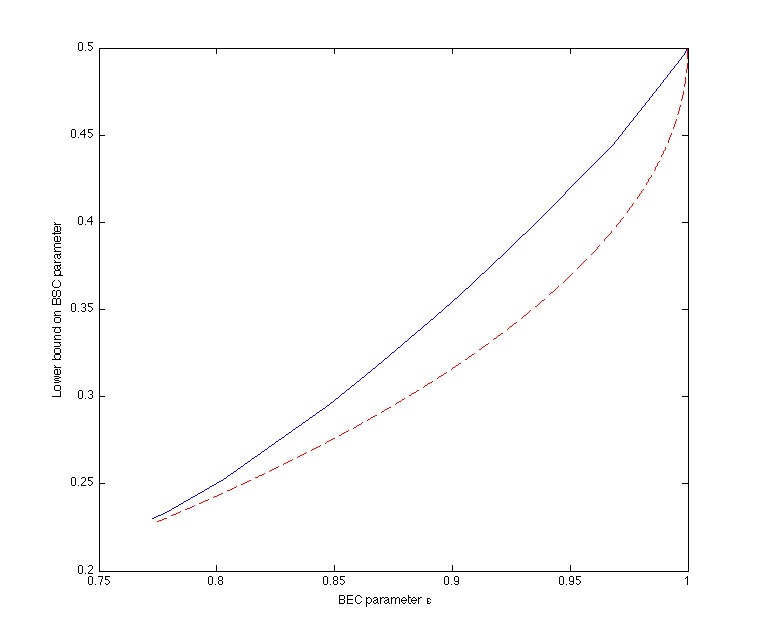}
\end{center}
\caption{The value of $p^*$ given in \eqref{eqnForFig} (upper curve), and $p=h^{-1}(\epsilon)\in[0,1/2]$ (lower curve) in terms of the BEC parameter $\epsilon$, for $\epsilon> \epsilon^*\approx 0.7729$.\color{black}}\label{FigNewRev}
\end{figure}

We show that the above converse is stronger than $1-\epsilon\geq 1-h(p)$ for all values of $\epsilon> \epsilon^*\approx 0.7729$, where $\epsilon^*\in (0,1]$ is the unique positive solution of $\epsilon^*=h(\epsilon^*)$. To prove this, let us choose $a=1-\epsilon$ and consider the following lower bound on $p^*$:
\begin{align}p^*\geq \underset{\alpha,\beta}{\min} \:{(\alpha\bar{\beta}+\bar{\alpha}\beta)},\label{eqnForFig}\end{align}
over pairs $(\alpha,\beta)$ satisfying \eqref{eqn:anewsb1}-\eqref{eqn:anewsb2} with $a=1-\epsilon$.  We have plotted $p^*$ in terms of $\epsilon$ in Fig.~\ref{FigNewRev}, and compared it with $p=h^{-1}(\epsilon)\in[0,1/2]$ that one gets from the trivial upper bound. However, to show the improvement analytically, \color{black} note that equations \eqref{eqn:anewsb2} and \eqref{eqn:anewsb3} imply that
$$h(\alpha)\geq \epsilon,
\qquad h(\beta)\geq \epsilon-h(\epsilon).$$
Since $\epsilon>\epsilon^*$,  we have $ \epsilon-h(\epsilon)>0$ and hence $h(\beta)>0$. Thus, $\beta\in[\beta^*, 1/2]$ where $\beta^*>0$ is the solution of
$h(\beta^*)=\epsilon-h(\epsilon)$. Therefore, $\alpha\bar{\beta}+\bar{\alpha}\beta\geq \alpha\bar{\beta}^*+\bar{\alpha}\beta^*$. We next have $h(\alpha)\geq \epsilon$, and thus, $\alpha\in [\alpha^*, 1/2]$ where $\alpha^*\in[0,1/2]$ is such that $h(\alpha^*)=\epsilon$. Hence, after solving the minimization over all $(\alpha,\beta)$, one ends up with 
$$p^*\geq \underset{\alpha,\beta}{\min} \:{(\alpha\bar{\beta}+\bar{\alpha}\beta)}=\alpha^*\bar{\beta}^*+\bar{\alpha^*}\beta^*.$$
Since for any $\alpha,\beta\in (0,1/2]$ we have $h(\alpha\bar{\beta}+\bar{\alpha}\beta)>h(\alpha)$, we find that $h(p^*)>\epsilon$, which results in a strictly better bound than $h(p)\geq \epsilon$ that one gets from the trivial upper bound. 
\end{example}

\subsection{Broadcast channel simulation}\label{sec:StrongBroadcast}

Consider the problem of simulating memoryless copies of the channel $p(y,z|x)$ given memoryless copies of $p(\tilde{y},\tilde{z}|\tilde{x})$ as depicted in Fig.~\ref{fig3}. In this setting, the input terminal observes i.i.d.\ copies of a source $X$ (taking values in finite sets $\mx$ and having joint pmf $p(x)$). The three terminals are provided with a shared randomness at rate $R$, denoted by random variable $\omega$, which is uniformly distributed over $[1:2^{nR}]$ and is independent of $X^n$.  An $(n, R)$ code consists of
\begin{itemize}
\item An (stochastic) encoder with conditional pmf's $q^{\enc}(\tilde{x}^{n}|x^{n}, \omega)$,
\item Two (stochastic) decoders with conditional pmf's $q^{\dec_1}(y^{n}|\tilde{y}^{n}, \omega)$ and $q^{\dec_1}(z^{n}|\tilde{z}^{n}, \omega)$.
\end{itemize}
Thus, the joint distribution induced by the code is as follows:
\begin{align}q(\hat{x}^n, \hat{y}^n, \hat{z}^n, y^n, z^n| \omega, x^n)=q^{\enc}(\tilde{x}^{n}|x^{n},\omega)(\prod_{i=1}^{n}p(\tilde{y}_i, \tilde{z}_i|\tilde{x}_i)\big)q^{\dec_1}(y^{n}|\tilde{y}^{n}, \omega)q^{\dec_2}(z^{n}|\tilde{z}^{n},\omega).
\end{align}

\begin{figure}[t]
\begin{center}
\begin{tikzpicture}[scale=1.2, thick]
\node at (-3,2) (jointdist){$X\sim p(x)$};
\node at (0,2.7)  (commonrdn1){$R$};
\node at (0,-0.7)  (commonrdn2){$R$};
\node  at ( 0,1) [rectangle,draw=black!50,fill=green!20!white,inner sep=8pt, minimum size=8mm] (ch) {$p(\tilde{y},\tilde{z}|\tilde{x})$};
\node at (-3,1) (input){$X^n$};
\node at (3,2) (output1){$Y^n$};
\node at (3,-0.1) (output2){$Z^n$};
\node at ( 2,2) [rectangle,draw=black!50,fill=green!20!white] (dec1) {$\mathcal{D}_{1}$};
\node at ( 2,-0.1) [rectangle,draw=black!50,fill=green!20!white] (dec2) {$\mathcal{D}_{2}$};
\node at (-2,1) [rectangle,draw=black!50,fill=green!20!white](enc) {$\mathcal{E}$};
\draw [->] (enc) -- node[above] {$\quad\tilde{X}^n$} (ch);
\draw [->] (ch)  -- node[above] {$\tilde{Y}^n$} (dec1);
\draw [->] (ch)  -- node[below] {$\tilde{Z}^n$} (dec2);
\draw [->,dashed]  (commonrdn1) to  (2,2.7) to (dec1);
\draw [->,dashed]  (commonrdn1) to  (-2,2.7) to (enc);
\draw [->,dashed]  (commonrdn2) to  (2,-0.7) to (dec2);
\draw [->,dashed]  (commonrdn2) to  (-2,-0.7) to (enc);
\draw [->] (input)--(enc);
\draw [->] (dec1)--(output1);
\draw [->] (dec2)--(output2);
\draw [->,dashed] (jointdist)--(input);
\end{tikzpicture}
\caption{Channel simulation over the Broadcast channel $p(\tilde{y},\tilde{z}|\tilde{x})$.}
\label{fig3}
\end{center}
\end{figure}
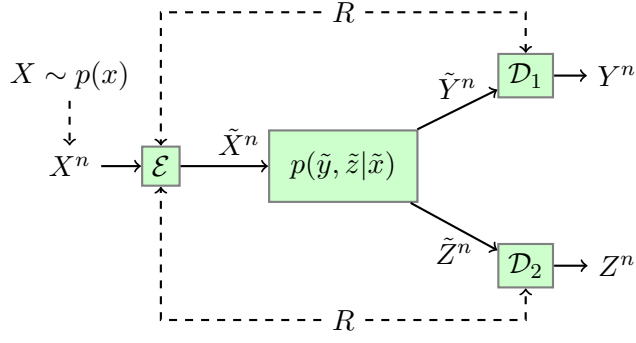

\begin{definition}
An input distribution-channel pair $\big(p(x), p(y,z|x)\big)$ is said to be in the admissible region of the channel-rate triple $\big(p(\tilde{y},\tilde{z}|\tilde{x}), R\big)$ if one can find a sequence of $(n, R)$ simulation codes whose induced joint distributions have marginal distributions $q(x^{n},y^{n}, z^{n})$ that satisfy
$$
\underset{n\rightarrow \infty}{\lim} \Big\| q(x^{n},y^{n}, z^n)-\overset{n}{\underset{i=1}{\prod}}p(x_{i},y_{i}, z_i)\Big\|_{1}=0.
$$
\end{definition}

We can now state our achievability bound for the broadcast channel simulation problem. 

\begin{theorem}[Inner Bound]\label{BCIn}
$\big(p(x), p(y,z|x)\big)$ is in the admissible region of the channel-rate pair $\big(p(\tilde{y},\tilde{z}|\tilde{x}), R\big)$ if
there exist $p(u, v, w, \tilde{x}|x)$, $p(y|\tilde{y}, u, w)$ and $p(z|\tilde{z}, v, w)$ such that $(X, U, V, W, \tilde{X}, \tilde{Y}, \tilde{Z}, Y, Z)$ is distributed according to
$$p(x, u, v, w, \tilde{x}, \tilde{y}, \tilde{z}, y, z)=p(x)p(u,v,w,\tilde{x}|x)
p(\tilde{y},\tilde{z}|\tilde{x}) p(y|\tilde{y}, u, w)p(z|\tilde{z}, v, w),$$ 
that has the given marginal $p(x,y,z)$, and satisfies
\begin{align}
I(WU;\tilde{Y})&>I(UW;X),\nonumber\\
R+I(WU;\tilde{Y})&>I(UW;XYZ),\nonumber\\
I(WV;\tilde{Z})&>I(WV;X),\nonumber\\
R+I(WV;\tilde{Z})&>I(WV;XYZ),\nonumber\\
I(WU;\tilde{Y})+I(WV;\tilde{Z})&>I(UW;X)+I(WV;X)+I(U;V|WX),\nonumber\\
2R+I(UW;\tilde{Y})+I(VW;\tilde{Z})&>I(UW;XYZ)+I(VW;XYZ)\nonumber\\
& \quad+I(U;V|WXYZ),\nonumber\\
\min\{I(W;\tilde{Y}),I(W;\tilde{Z})\}+I(U;\tilde{Y}|W)+I(V;\tilde{Z}|W)&>I(W;X)+I(U;X|W)+I(V;X|W)\nonumber\\
&\quad +I(U;V|WX),\nonumber\\
R+\min\{I(W;\tilde{Y}),I(W;\tilde{Z})\}+I(U;\tilde{Y}|W)+I(V;\tilde{Z}|W)&>I(W;XYZ)+I(U;XYZ|W)\nonumber\\
&\quad+I(V;XYZ|W)+I(U;V|WXYZ),\nonumber\\
R+I(W;YZ|X)+I(UW;\tilde{Y})+I(VW;\tilde{Z})&>I(UW;XYZ)+I(VW;XYZ)\nonumber\\
&\quad +I(U;V|WXYZ).\label{thm:bcc}
\end{align}
\end{theorem}


\section{Exact channel simulation }\label{Zsec}

Recall that we say a channel $p(y|x)$ can be simulated exactly with a channel $p(\tilde y| \tilde x)$ if there are $n$, encoding map $q^{\enc}(\tilde{x}^{n}|x^{n},\omega)$, and decoding map $q^{\dec}(y^{n}|\tilde{y}^{n},\omega)$ such that the induced distribution given by~\eqref{eq:induced-dist-2} satisfies
\begin{align}\label{eq:n-sim-exact-4}
q(y^n|x^n)=\prod_{i=1}^np(y_i|x_i).
\end{align} 
In this section we state our results about the channel simulation problem in the exact model.

\begin{theorem}\label{exact-converse}
A channel $p(y|x)$ can be simulated from $p(\tilde{y}|\tilde{x})$ in the exact model with infinite  shared randomness \color{black} only if for any $\alpha>0$ we have
\begin{align*}
C_\alpha(p(y|x))&\leq C_\alpha(p(\tilde y|\tilde x)),\\
\mathsf{Diam}_\alpha(p(y|x))&\leq \mathsf{Diam}_\alpha(p(\tilde y|\tilde x)),
\end{align*}
where $C_\alpha$ denotes the $\alpha$-capacity defined in~\eqref{eq:alpha-capacity}, and  
$$\mathsf{Diam}_\alpha(p(y|x))\triangleq\max_{p(x), q(x)}D_{\alpha}(p(y)\|q(y)).$$
\end{theorem}

\begin{remark}
To the best of our knowledge, we define the notion of channel diameter $\mathsf{Diam}_\alpha(p(y|x))$ for the first time. 
\\
Theorem \ref{exact-converse} implies that the exact simulation of $p(y|x)$ with a noiseless link of rate $R$ with infinite shared randomness is possible only if $C_\infty(p(y|x))\leq R.$
In fact, as shown in \cite{Winter1} (see \cite{exactcoordination} for a discussion), 
$C_\infty(p(y|x))\leq R$
is a sufficient and necessary condition for exact simulation of $p(y|x)$ with a noiseless link of rate $R$.
\end{remark}
\begin{remark} As it becomes clear from the proof of Theorem \ref{exact-converse}, any function $\Phi(p(y|x))$ that satisfies additivity and data processing properties, namely,
$$\Phi\big(\prod_{i=1}^np(y_i|x_i)\big)=n \Phi\big(p(y|x)\big)$$
and 
$$\Phi(p(y|x))\geq \Phi(p(z|x)), \qquad \emph{if } ~~~p(z|x)=\sum_{y}p(y|x)p(z|y)$$
can be used instead of $C_\alpha$ and $\mathsf{Diam}_\alpha$ above to write an infeasibility result when there is no shared randomness. Furthermore, if  $\Phi(p(y|x))$ is quasi-convex in $p(y|x)$, one can claim the infeasibility result even when there is infinite shared randomness. 
The above particular choices of $\Phi(p(y|x))$ as $C_\alpha$ and $\mathsf{Diam}_\alpha$ are shown later to be particularly useful in finding \emph{tight bounds} for a few examples we consider. But it is possible to find other choices for $\Phi(p(y|x))$: for instance, $E_{p(y|x)}(R)$, the reliability function of $p(y|x)$ at a given rate $R$ satisfies the additivity and data processing inequalities by definition, and is quasi-convex since shared randomness does not increase the optimal error exponent. 
\color{black}
\end{remark}

Using the above theorem, we show that the degradation strategy of Fig.~\ref{fig2news4} is optimal in the exact model for simulating a BSC channel from a BEC channel (for any amount of shared randomness).

\begin{theorem}\label{BSCCAP}
For any amount of shared randomness, the $\BSC(p)$ channel with parameter $p\in[0,1/2]$ can be exactly simulated from the $\BEC(\epsilon)$ if and only if $p\geq \epsilon/2$. In particular, shared randomness is not helpful in this exact simulation problem.
\end{theorem}
\color{black}

A channel is called binary-input binary-output (BIBO) if both the input and the output of the channel are a single bit.  A BIBO channel can be characterized with two parameters $r=p_{Y|X}(0|0)$ and $s=p_{Y|X}(0|1)$ as depicted in Fig.~\ref{AssBSCZ}. The following theorem studies BIBO channels that can be exactly simulated from another BIBO channel. 
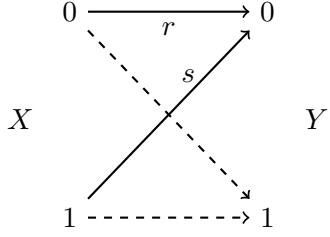
\begin{figure}
\begin{center}
\begin{tikzpicture}[scale=1.3,thick]
\node at (-1.5,0.9){$X$};
\node at (1.5,0.9){$Y$};
\node at (-1,2)(test1){$0$};
\node at (1,-0.1)(test2){$1$};
\node at (-1,-0.1) (test3){$1$};
\node at (1,2)(test4){$0$};
\node at (0.2,1.35){$s$};
\draw [->,dashed] (test1)--(test2);
\draw [->] (test3)--(test4);
\draw [->,dashed] (test3)--(test2);
\draw [->](test1)--node[below]{$r$}(test4);
\end{tikzpicture}
\caption{A BIBO channel $p(y|x)$ with parameters $(r,s)$.}
\label{AssBSCZ}
\end{center}
\vskip -0.5cm
\end{figure}

\begin{theorem}\label{BSCCAP2} Depending on whether we have shared randomness or not, we have
\begin{itemize}
\item (Infinite shared randomness): A channel with parameters $(r,s)$ can be simulated exactly  from a channel with parameters $(r',s')$ if and only if $(r,s)$ is in the convex polygon with six vertices $(0,0), (1,1), (r',s'), (s',r'), (\bar{s}', \bar{r}'),$ and  $(\bar{r}', \bar{s}')$ where $\bar{s}'=1-s'$ and $\bar{r}'=1-r'$ (see Fig.~\ref{InnerBoundExactFig}).
\item (No shared randomness):
A BIBO channel with parameters $(r,s)$ can be simulated from another BIBO channel with parameters $(r',s')$ in the exact model  if $(r,s)$ is in the union of two parallelogram with vertices $(0,0), (r',s'), (\bar{r}', \bar{s}'), (1,1)$ and
$(0,0), (s',r'), (\bar{s}', \bar{r}'), (1,1)$  (see Fig.~\ref{InnerBoundExactFig}). Conversely, if a channel with parameters $(r,s)$ can be simulated exactly from a channel with parameters $(r',s')$, then $(r,s)$ has to be in the convex polygon with six vertices $(0,0), (1,1), (r',s'), (s',r'), (\bar{s}', \bar{r}'),$ and  $(\bar{r}', \bar{s}')$ depicted in Fig.~\ref{InnerBoundExactFig}. In other words, the outer bound  is the convex hull of the given achievable inner bound.
\end{itemize}
\end{theorem}

\color{black}

\begin{figure}
\begin{center}
\includegraphics[scale=0.5]{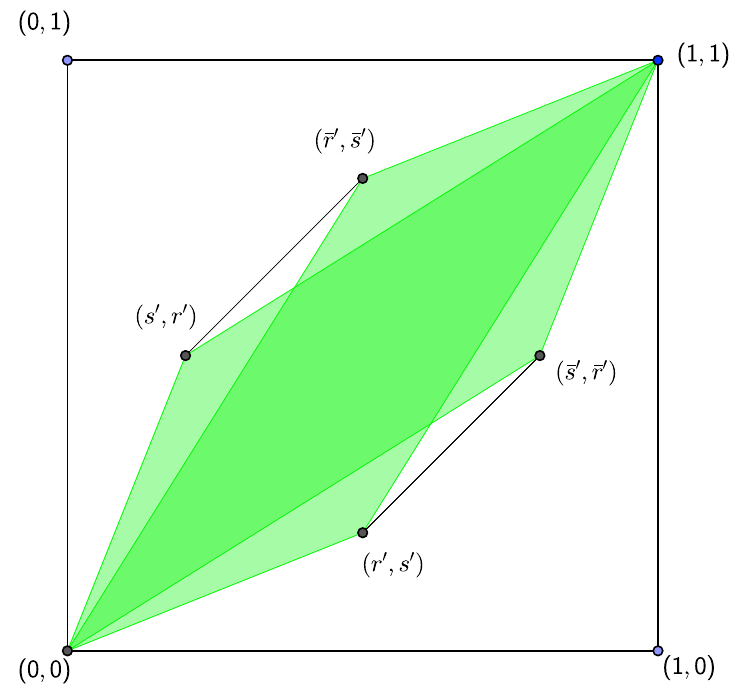}
\caption{The inner bound (union of two parallelogram) and outer bound (convex polygon) for exact channel simulation of BIBO  channel from another BIBO channel when there is no shared randomness. The convex polygon bound is the simulation region when there is infinite shared randomness.}
\label{InnerBoundExactFig}
\end{center}
\end{figure}


\section{Proofs}\label{sec:proofs}

This section is devoted to the proofs of the results stated in the previous two sections. 

\subsection{Point to point channel: inner bound}
\begin{proof}[Proof of Theorem \ref{Pt2PtIn}]
We apply the OSRB technique of \cite{OSRB} to prove the theorem.  Our proof consists of three parts. In the first part we introduce two protocols, A and B, each of which induces a pmf on a certain set of random variables. Protocol A has the desired i.i.d.\ property on $X^{n}$ and $Y^{n}$, but leads to no concrete coding algorithm. However, Protocol B is suitable for construction of a code, with one exception: Protocol B is assisted with an extra common randomness that does not really exist in the model. In the second part of the proof we find conditions on $R$ implying that two certain induced distributions are almost identical. In the third part of the proof, we eliminate the extra common randomness given to  Protocol B without significantly disturbing the pmf induced on the desired random variables $(X^n,Y^{n})$. This makes Protocol B useful for code construction.\\

\noindent
\textbf{Part (1)}: We define two protocols each of which induces a joint pmf on random variables of the corresponding protocol.\\

\noindent
\emph{Protocol A [Not useful for coding].} Let $(U^{n},X^{n},\tilde{X}^{n},\tilde{Y}^n,Y^{n})$ be $n$ i.i.d.\ copies of the joint pmf $p(u,x,\tilde{x},\tilde{y},y)$. Consider the following construction:
\begin{itemize}
\item To each $u^{n}\in \mathcal U^n$ assign two random bin indices $g\in[1:2^{n\tilde{R}}]$ and  $\omega\in[1:2^{nR}]$.
\item Consider a Slepian-Wolf decoder for estimating $\hat{u}^{n}$ from $(\omega,g,\tilde{y}^{n})$. Here we are considering $\tilde{y}^{n}$ as side information and $\omega, g$ as the random bins of the source $u^n$ that we want to decode.
\end{itemize}

The rate constraints on $R, \tilde R$ for the success of this decoder will be imposed
later, although this decoder can be conceived even when there is no guarantee of successful decoding. We denoted the random pmf induced by the random binning and the Slepian-Wolf decoder by $Q$ and $Q^{\mathsf{sw}}$, respectively. We then obtain the following joint distribution: 
\begin{align}
Q(x^{n},u^{n},\hat{u}^n,\tilde{x}^{n},\tilde{y}^n, y^{n},g,\omega)&=p(x^{n})p(u^{n},\tilde{x}^{n}|x^{n})p(\tilde{y}^{n}|\tilde{x}^{n})\nonumber\\
&\quad\times Q(g,\omega |u^{n})Q^{\mathsf{sw}}(\hat{u}^{n}|g,w,\tilde{y}^{n})p(y^{n}|\tilde{y}^{n},u^{n})\nonumber\\
&=p(x^{n})Q(g,\omega |x^{n})Q(u^{n},\tilde{x}^{n}|x^{n},g,\omega)p(\tilde{y}^{n}|\tilde{x}^{n})\nonumber\\
&\quad \times Q^{\mathsf{sw}}(\hat{u}^{n}|g,w,\tilde{y}^{n})p(y^{n}|\tilde{y}^{n},u^{n}).
\label{eq:fist-part-A}
\end{align}
Note that we have used capital $Q$ for random pmf's in the above equation.\\

\noindent
\emph{Protocol B [Useful for coding after removing an extra common randomness]}. In this protocol we assume that Alice and Bob have access to the common randomness $\omega$ and an extra common randomness $G$, where $G$ is mutually independent of $X^n$ and $\omega$. We assume that $G$ is distributed uniformly over the set  $[1:2^{n\tilde{R}}]$. Now we use the following protocol:
\begin{itemize}
 \item First, Alice having $(g,\omega,x^{n})$ generates $u^{n}$ according to the pmf $Q(u^{n}, \tilde{x}^{n}|g,\omega,x^{n})$ of Protocol A, and sends $\tilde{x}^{n}$ over the channel to Bob. Then Bob receives $\tilde{y}^n$. Having $(g,\omega,\tilde{y}^{n})$, Bob uses the Slepian-Wolf decoder of protocol A to generate $\hat u^n$ as a estimation of $u^{n}$.

\item Having $(\hat{u}^{n},\tilde{y}^{n})$, Bob generates $Y^{n}$ according to $p(y^{n}|\tilde{y}^{n}, \hat{u}^{n})=\prod_{i=1}^n
p_{Y|\tilde{Y}U}(y_i|\tilde{y}_i,\hat{u}_i)$.
\end{itemize}

The random pmf induced by the protocol, denoted by $\hat{Q}$, is equal to
\begin{align}\label{eq:fist-part-B}
\hat{Q}(x^{n},u^{n},\hat{u}^n,\tilde{x}^{n},\tilde{y}^n,y^{n},g,\omega)
&=p^{\mathsf u}(\omega)p^{\mathsf u}(g)p(x^{n})Q(u^{n},\tilde{x}^{n}|g,\omega, x^{n})p(\tilde{y}^n|\tilde{x}^n)\nonumber\\
&\quad \times  Q^{\mathsf{sw}}(\hat{u}^{n}|\omega ,g,\tilde{y}^{n})p(y^{n}|\tilde{y}^{n},\hat{u}^{n}),
\end{align}
where $p^{\mathsf u}$ denotes the uniform distribution.\\

\noindent\textbf{Part (2)}: In this part we put  sufficient conditions under which the induced distributions $Q$ and $\hat Q$ given by~\eqref{eq:fist-part-A} and~\eqref{eq:fist-part-B} are approximately the same. 
The first step is to observe that $g$ and $\omega$ are the bin indices of $u^n$. Substituting $T=2$, $X_1\leftarrow U$, $X_2\leftarrow U$, $Z\leftarrow X$, $b_1 \leftarrow g$ and $b_2 \leftarrow \omega$ in Theorem~1 of~\cite{OSRB}, implies that if
\begin{align}
R+\tilde{R}<H(U|X),\label{ist1}
\end{align}
then there exists $\epsilon^{(n)}_{0}\rightarrow 0$ such that $$\mathbb{E}\|Q(g,\omega|x^{n})-p^{\mathsf u}(\omega)p^{\mathsf u}(g)\|_1 \leq \epsilon^{(n)}_{0}.$$\color{black} Observe that $p^{\mathsf u}(\omega)p^{\mathsf u}(g)=\hat{Q}(g,\omega)$.
This implies that the joint pmf of all random variables, excluding $y^n$, of the two protocols are close in total variation distance, i.e, 
\begin{align}\mathbb{E}\|
\hat{Q}(x^{n},u^{n}&,\hat{u}^n,\tilde{x}^{n},\tilde{y}^n,g,\omega)
-Q(x^{n},u^{n},\hat{u}^n,\tilde{x}^{n},\tilde{y}^n,g,\omega)\|_1\leq \epsilon^{(n)}_{0}.
\label{eqn:P10}
\end{align}\color{black}
To ensure the above equation with $y^n$ included, we begin by investigating conditions that make the Slepian-Wolf decoder of Protocol A succeed with high probability. By the Slepian-Wolf theorem as long as
\begin{align}
R+\tilde{R}>H(U|\tilde{Y})\label{ist2}
\end{align}
holds, we have
\begin{align}
Q(x^{n},u^{n},\hat{u}^n&,\tilde{x}^{n},\tilde{y}^n,g,\omega)\overset{\epsilon^{(n)}_{1}}{\approx}Q(x^{n},u^{n},\tilde{x}^{n},\tilde{y}^n,g,\omega)\boldsymbol{1}\{\hat{u}^n=u^{n}\}.\label{eqn:P11}
\end{align}
for some vanishing sequence $\epsilon^{(n)}_{1}$. Then using equations~\eqref{eqn:P10} and~\eqref{eqn:P11} we can apply Lemma~3 of~\cite{OSRB} to write
\begin{align}
\hat{Q}(x^{n},u^{n},\hat{u}^n&,\tilde{x}^{n},\tilde{y}^n,g,\omega)\overset{\epsilon ^{(n)}_{0}+\epsilon^{(n)}_{1}}{\approx}Q(x^{n},u^{n},\tilde{x}^{n},\tilde{y}^n,g,\omega)\boldsymbol{1}\{\hat{u}^n=u^{n}\}.
\end{align}
Moreover, the third part of Lemma~3 of~\cite{OSRB} implies that
\begin{align}
\hat{Q}(x^{n},u^{n},\hat{u}^n,\tilde{x}^{n},\tilde{y}^n,g,\omega)p(y^n|\hat{u}^{n},\tilde{y}^{n})
&\overset{\epsilon ^{(n)}_{0}+\epsilon^{(n)}_{1}}{\approx}
Q(x^{n},u^{n},\tilde{x}^{n},\tilde{y}^n,g,\omega)\boldsymbol{1}\{\hat{u}^n=u^{n}\}p(y^n|\hat{u}^{n},\tilde{y}^{n})\nonumber\\
&~\quad=Q(x^{n},u^{n},\tilde{x}^{n},\tilde{y}^n,g,\omega)\boldsymbol{1}\{\hat{u}^n=u^{n}\}p(y^n|u^n, \tilde{y}^n)\nonumber\\
&~\quad =Q(x^{n},u^{n},\tilde{x}^{n},\tilde{y}^n,g,\omega)p(y^n|u^n, \tilde{y}^n)\boldsymbol{1}\{\hat{u}^n=u^{n}\}\nonumber\\
&~\quad =Q(x^{n},u^{n},\tilde{x}^{n},\tilde{y}^n,y^n,g,\omega)\boldsymbol{1}\{\hat{u}^n=u^{n}\}.\nonumber
\end{align}
Then, by part~1~(second item) of Lemma~3 of~\cite{OSRB} we have
$
\hat{Q}(g,x^n,y^n)\overset{\epsilon ^{(n)}_{0}+\epsilon^{(n)}_{1}}{\approx}Q(g,x^n,y^n)$.
Note, in particular, that the marginal pmf on $(X^n,Y^n)$ of $Q$ is equal to $p(x^n,y^n)$ implying that $\hat{Q}(x^n, y^n)$ is within $\epsilon ^{(n)}_{0}+\epsilon^{(n)}_1$ distance of $p(x^n, y^n)$.

To summarize, assuming~\eqref{ist1} and~\eqref{ist2}, and having access to common randomness $\omega, G$, Alice and Bob can simulate the channel $p(y|x)$ using the channel $p(\tilde y| \tilde x)$ according to Protocol B. As discussed above, with high probability $\hat u^n$ generated by Bob would be equal to $u^n$, and the final pmf induced on $(X^n, Y^n)$ would be close to the desired pmf $p(x^n, y^n)$.\\

\noindent\textbf{Part (3)}:
In the above protocol we assumed that Alice and Bob have access to an extra randomness $G$
which is not present in the model. To eliminate this extra common randomness, we will fix a particular instance $g$ of $G$ and show that the same protocol works even if we fix $G=g$. To prove this note that by letting $G=g$, the induced pmf $\hat{Q}(x^n,y^n)$ changes to the conditional pmf $\hat{Q}(x^n,y^n|g)$. But if $G$ is almost independent of $(X^n,Y^n)$, the conditional pmf $\hat{Q}(x^n,y^n|g)$ would be close to the desired distribution as well. To
obtain the independence, we again use Theorem~1 of~\cite{OSRB}. Substituting $T = 1$, $X_1\leftarrow U$ and
$Z \leftarrow XY$ in Theorem 1 of \cite{OSRB}, we find that if
\begin{align}
\tilde{R}<H(U|XY),
\label{ist3}
\end{align}
then $Q(x^n,y^n,g)\overset{\epsilon^{(n)}_{2}}{\approx}p^{\mathsf u}(g)p(x^n,y^n)$, for some vanishing $\epsilon ^{(n)}_{2}$. Thus, by triangular inequality for total variation distance, we have $\hat{Q}(x^n,y^n,g)\overset{\epsilon ^{(n)}}{\approx}p^{\mathsf u}(g)p(x^n,y^n)$, where $\epsilon^{(n)}=\sum_{i=0}^{2}\epsilon^{(n)}_{i}$. From the definition of total variation distance for random pmf's, the average of the total variation distance between $\hat{q}(x^n,y^n,g)$ and $p^{\mathsf u}(g)p(x^n,y^n)$ over all random binnings is small. Thus, there exists a fixed binning with the corresponding pmf $\hat{q}$ such that $\hat{q}(x^n,y^n,g)\overset{\epsilon^{(n)}}{\approx}p^{\mathsf u}(g)p(x^n,y^n)$. Next, Lemma~3 of~\cite{OSRB} guarantees the existence of an instance $g$ such that $\hat{p}(x^n,y^n|g)\overset{2\epsilon^{(n)}}{\approx}p(x^n,y^n)$. Then the extra shared randomness $G$ can be eliminated by fixing it to be $G=g$. 

Finally, observe that the conditions of~\eqref{thm1region} are seen to be equivalent to~\eqref{ist1}, \eqref{ist2} and~\eqref{ist3} after eliminating $\tilde{R}$ using Fourier-Motzkin elimination.
\end{proof}

\begin{proof}[Proof of Theorem \ref{Subth}]
We would like to prove that to simulate a BSC channel with a non-uniform input pmf from a BEC channel, we can do better than $p=\frac{\epsilon}{2}$ (obtained by a degradation scheme). To be more precise, let $\Bern(q)$ be the Bernoulli distribution with parameter $q$. 
We show that $(\Bern(q), \BSC(\frac{t\epsilon}{2}))$ is in the admissible region of $(\BEC(\epsilon),0)$ for any $t\in(0,1]$ such that 
\begin{align}t\left[\epsilon+(1-\epsilon)(1-h(q))+h\left(\frac{\epsilon}{2}\right)\right]-\epsilon-h\left(\frac{t\epsilon}{2}\right)\geq 0.
\label{eqneqn123}\end{align}
Indeed if $q\neq 1/2$, this inequality strictly holds for $t=1$. So for any $q\neq 1/2$, one can find $t<1$ so that this inequality is still valid. This would demonstrate the sub-optimality of a degradation scheme for non-uniform input distributions.

We use Theorem~\ref{Pt2PtIn} to prove the above claim. For this we need to specify the joint pmf of random variables $X, Y, U, \tilde{X}, \tilde{Y}$ as follows:
\begin{itemize}
\item Let $X$ to be distributed according to $\Bern(q)$.
\item Assume that $\tilde{X}$ is uniform over $\{0,1\}$ and independent of $X$.
\item Let $W$ be $\Bern(t)$ and independent of $(X, \tilde X)$. 
\item Define $U$ as follows: let $U=(W, K)$ where $K=(X, \tilde{X})$ if $W=0$, and $K=X+\tilde X(\mod 2)$ if $W=1$. 
\item To specify $p(y|\tilde{y}, u)$ we proceed as follows: 
\begin{itemize}
\item If $W=0$, we let $Y=X$; note that in this case $X$ is a part of $U$. 
\item If $W=1$, we look at $\tilde Y$; if it is the erasure flag, we choose $Y$ uniformly at random. Otherwise, $\tilde Y=\tilde X$, so we may let $Y=\tilde X + K(\mod 2)=X$. 
\end{itemize}
\end{itemize}

This procedure induces the following distribution on $(X, Y)$: $X$ is chosen according to $\Bern(q)$; with probability $(1-t)+t(1-\epsilon)$ we have $Y=X$, and $Y$ with probability $t\epsilon$ is chosen uniformly at random (and independent of $X$). This is equivalent with the $\BSC(p)$ channel with $\Bern(q)$ input distribution where
$$p=\frac{t\epsilon}{2}.$$

We now need to verify~\eqref{thm1region} for $R=0$, i.e., $I(U; \tilde Y)> I(U; XY)\geq I(U; X)$.
We have 
$$I(U;\tilde{Y})=(1-\epsilon)I(U;\tilde{X})=(1-\epsilon)\big[1-t+t(1-h(q))\big].$$
Next,
$$I(U;XY)=I(WK;XY)=I(W;XY)+I(K;XY|W).$$
Moreover,  
\begin{align*}I(W;XY)&=H(XY)-H(XY|W)\\&=H(XY)-(1-t)H(X)-tH(XY|W=1)\\&=h(q)+h(p)-(1-t)h(q)-t\cdot h(q)-t\cdot h(\frac{\epsilon}{2})
\\&=h(p)-t\cdot h(\frac{\epsilon}{2})
\\&=h(\frac{t\epsilon}{2})-t\cdot h(\frac{\epsilon}{2}),
\end{align*}
and
$$I(K;XY|W)=(1-t)I(X\tilde{X};XX)+tI(X+\tilde{X};XY)=1-t,$$
where we use the fact that $\tilde{X}$ is independent of $(X, Y)$ if $W=1$.
Therefore,
\begin{align*}
I(U;\tilde{Y})-I(U;XY)&=(1-\epsilon)\big[1-t+t(1-h(q))\big]-h(\frac{t\epsilon}{2})+t\cdot h(\frac{\epsilon}{2})-1+t
\\&=t\big[\epsilon+(1-\epsilon)(1-h(q))+h(\frac{\epsilon}{2})\big]-\epsilon-h(\frac{t\epsilon}{2}),
\end{align*}
which is positive by assumption~\eqref{eqneqn123}.
\end{proof}

\subsection{Point to point channel: outer bound}

\begin{proof}[Proof of Theorem \ref{converseThm}]
Take an encoding map $q^{\enc}(\tilde x^n|x^n, \omega)$ and a decoding map $q^{\dec}(y^n| \tilde y^n, \omega)$ with the induced distribution $q(x^n, y^n, \tilde x^n, \tilde y^n, \omega)$ as described in~\eqref{eq:induced-dist-2} such that  
\begin{align}\label{eq:conversethm-eps-2}
\Big\| q(x^{n},y^{n})-\overset{n}{\underset{i=1}{\prod}}p(x_{i})p(y_{i}|x_i)\Big\|_{1}\leq \epsilon.
\end{align}
We have the Markov chains 
\begin{align*}
X^n\rightarrow \omega\tilde{X}^n\rightarrow\omega\tilde{Y}^n\rightarrow Y^n,
\end{align*}
and
\begin{align*}
\omega\rightarrow\tilde{X}^n\rightarrow \tilde{Y}^n.
\end{align*}
Moreover, $\omega$ is independent of $X^n$. Therefore,
$I(X^n;Y^n|\omega)\leq I(\tilde{X}^n;\tilde{Y}^n|\omega)$. On the other hand, $I(X^n;Y^n|\omega)=I(X^n;Y^n\omega)\geq I(X^n;Y^n)$, and
$ I(\tilde{X}^n;\tilde{Y}^n|\omega)\leq  I(\omega\tilde{X}^n;\tilde{Y}^n)=I(\tilde{X}^n;\tilde{Y}^n)$. As a result, $I(X^n;Y^n)\leq I(\tilde{X}^n;\tilde{Y}^n)$. 

To proceed let
$$f_1=\frac{1}{n}\sum_{i=1}^nI(X_{[1:i-1]}Y_{[1:i-1]};X_i Y_i).$$
Since by~\eqref{eq:conversethm-eps-2} the induced distribution on $(X^n, Y^n)$ by the code is almost i.i.d., $f_1$ should be small. Indeed, Lemma~\ref{cuflem} below shows that $f_1$ vanishes as $\epsilon$ goes to zero.
We can then write
\begin{align}
\sum_{i=1}^n I(X_i; Y_i)&= \sum_{i=1}^n H(X_i)+H(Y_i)-H(X_i, Y_i)
\nonumber\\&\leq \sum_{i=1}^n H(X_i)+H(Y_i)-H(X_i, Y_i|X_{[1:i-1]},Y_{[1:i-1]})
\nonumber\\&= \sum_{i=1}^n H(X_i|X_{[1:i-1]})+H(Y_i|Y_{[1:i-1]})-H(X_i, Y_i|X_{[1:i-1]},Y_{[1:i-1]})+I(X_i;X_{[1:i-1]})+I(Y_i;Y_{[1:i-1]})
\nonumber\\& =H(X^n)+H(Y^n)-H(X^n, Y^n)+\sum_{i=1}^n I(X_i;X_{[1:i-1]})+I(Y_i;Y_{[1:i-1]})
\nonumber\\&\leq I(X^n;Y^n)+2nf_1
\nonumber\\&\leq I(\tilde{X}^n;\tilde{Y}^n)+2nf_1
\nonumber\\&\leq \sum_{i=1}^nI(\tilde{X}_i;\tilde{Y}_i)+2nf_1,\label{eq:equiv-1}
\end{align}
where in the last step we used the familiar expansion of mutual information for memoryless channels used in the converse proof of a point-to-point channel.

For $\tilde{X}_i, \tilde{Y}_i$, let $\tilde{U_i}$ be the random variable such that $\tilde{X}_i\rightarrow\tilde{U}_i\rightarrow\tilde{Y}_i$ forms a Markov chain and $\beta I(\tilde{U}_i;\tilde{X}_i\tilde{Y}_i)+\gamma I(\tilde{U}_i;\tilde{X}_i)+\theta I(\tilde{U}_i;\tilde{Y}_i)$ reaches its minimum. Assume that $\tilde{U_i}$ is constructed to be conditionally  independent of other variables given $\tilde{X}_i, \tilde{Y}_i$. \color{black} Then, we obtain a random variable $\tilde{U}^n=(\tilde{U}_1, \tilde{U}_2, \cdots, \tilde{U}_n)$ such that 
$$\tilde{X}^n\rightarrow\tilde{U}^n\rightarrow\tilde{Y}^n,$$ 
forms a Markov chain, and that $q(\tilde{u}^n|\tilde{x}^n)$ and $q(\tilde{y}^n|\tilde{u}^n)$ are product channels. The joint pmf of all random variables factorizes as
\begin{align}
q(x^n, \omega, \tilde x^n, \tilde u^n, \tilde y^n, y^n)&=
p(x^n)p(\omega)q^{\enc}(\tilde x^n|x^n, \omega)
q(\tilde u^n|\tilde x^n)q(\tilde y^n|\tilde u^n)q^{\dec}(y^n|\tilde y^n,\omega)\label{eq:revision2}
\\&=q(x^n,\omega,\tilde x^n, \tilde u^n)q(\tilde y^n|\tilde u^n)q^{\dec}(y^n|\tilde y^n,\omega)\nonumber\\
& = q(x^n,\omega,\tilde x^n, \tilde u^n)q(y^n\tilde y^n|\tilde u^n, \omega).\nonumber
\end{align}
This shows that $X^n\tilde{X}^n\rightarrow\tilde{U}^n\omega\rightarrow \tilde{Y}^nY^n$ forms a Markov chain. In particular, we conclude that 
$$X^n\rightarrow\tilde{U}^n\omega\rightarrow Y^n,$$ 
forms a Markov chain.


Let $U_i=(\tilde{U}^n,\omega)$. Then $X^n\rightarrow U_i\rightarrow Y^n$, and hence $X_i\rightarrow U_i\rightarrow Y_i$ form Markov chains. Next, we have
\begin{align}
\sum_{i=1}^n I(U_i;X_i, Y_i)&\leq
\sum_{i=1}^nI(U_iX_{[1:i-1]}Y_{[1:i-1]};X_i, Y_i)
\nonumber\\&=\sum_{i=1}^nI(X_{[1:i-1]}Y_{[1:i-1]};X_i, Y_i)+\sum_{i=1}^nI(U_i;X_i, Y_i|X_{[1:i-1]},Y_{[1:i-1]})
\nonumber\\&=nf_1+\sum_{i=1}^nI(\tilde{U}^n,\omega;X_i, Y_i|X_{[1:i-1]},Y_{[1:i-1]})
\nonumber\\&=nf_1+I(\tilde{U}^n\omega; X^n, Y^n)
\nonumber\\&\leq nf_1+H(\omega)+I(\tilde{U}^n;X^n, Y^n)
\nonumber\\ & \leq nf_1+H(\omega)+I(\tilde{U}^n;\tilde{X}^n, \tilde{Y}^n)
\nonumber\\ & \leq nf_1+nR+H(\tilde{U}^n)-H(\tilde{U}^n|\tilde{X}^n, \tilde{Y}^n)
\nonumber\\ & \leq nf_1+nR+\sum_{i=1}^nH(\tilde{U}_i)-H(\tilde{U}_i|\tilde{X}^n, \tilde{Y}^n, \tilde{U}_{[1:i-1]})
\nonumber\\ & = nf_1+nR+\sum_{i=1}^nH(\tilde{U}_i)-H(\tilde{U}_i|\tilde{X}_i, \tilde{Y}_i)\label{eq:equiv-2abc}
\\ & = nf_1+nR+\sum_{i=1}^nI(\tilde{U}_i;\tilde{X}_i, \tilde{Y}_i),\label{eq:equiv-2}
\end{align}
where equation \eqref{eq:equiv-2abc} follows from the fact that $p(\tilde{u}^n|\tilde x^n, \tilde y^n)=\prod_{i=1}^n p(\tilde{u}_i|\tilde x_i, \tilde y_i)$. This fact follows from the factorization
 $$p(\tilde x^n, \tilde u^n, \tilde y^n)=p(\tilde x^n)\prod_{i=1}^np(\tilde u_i, \tilde y_i|\tilde x_i).$$
Further, we have
\begin{align}
\sum_{i=1}^n I(U_i; Y_i)&\leq
\sum_{i=1}^nI(U_i,Y_{[1:i-1]};Y_i)
\nonumber\\&=\sum_{i=1}^nI(Y_{[1:i-1]}; Y_i)+\sum_{i=1}^nI(U_i;Y_i|Y_{[1:i-1]})
\nonumber\\&\leq nf_1+\sum_{i=1}^nI(\tilde{U}^n,\omega; Y_i|Y_{[1:i-1]})
\nonumber\\&=nf_1+I(\tilde{U}^n\omega; Y^n)
\nonumber\\ & \leq nf_1+2H(\omega)+I(\tilde{U}^n; \tilde{Y}^n)\label{eq:revision1}
\\ & \leq nf_1+2nR+H(\tilde{Y}^n)-H(\tilde{Y}^n|\tilde{U}^n)
\nonumber\\ & \leq nf_1+2nR+\sum_{i=1}^nH(\tilde{Y}_i)-H(\tilde{Y}_i|\tilde{U}^n, \tilde{Y}_{[1:i-1]})
\nonumber\\ & = nf_1+2nR+\sum_{i=1}^nH(\tilde{Y}_i)-H(\tilde{Y}_i|\tilde{U}_i) \label{eq:equiv-3news}
\\ & = nf_1+2nR+\sum_{i=1}^nI(\tilde{U}_i;\tilde{Y}_i).
\label{eq:equiv-3}
\end{align}
where equation \eqref{eq:revision1} holds because
\begin{align*}I(\tilde{U}^n\omega; Y^n)&\leq H(\omega)+I(\tilde{U}^n;Y^n|\omega)
\\&
\leq H(\omega)+I(\tilde{U}^n;\tilde{Y}^nY^n|\omega)
\\&= H(\omega)+I(\tilde{U}^n;\tilde{Y}^n|\omega)
\\&\leq  2H(\omega)+I(\tilde{U}^n;\tilde{Y}^n)
\end{align*}
\color{black}
and  equation \eqref{eq:equiv-3news} is due to the fact that $p(\tilde y^n|\tilde u^n)=\prod_{i=1}^np(\tilde y_i|\tilde u_i)$.
Following similar steps, using the fact that $p(\tilde u^n|\tilde x^n)=\prod_{i=1}^np(\tilde u_i|\tilde x_i)$ and 
\begin{align*}I(\tilde{U}^n\omega; X^n)&\leq H(\omega)+I(\tilde{U}^n;X^n)
\\&
\leq H(\omega)+I(\tilde{U}^n;\tilde{X}^nX^n)
\\&=H(\omega)+I(\tilde{U}^n;\tilde{X}^n)
\end{align*}\color{black}
which holds because of equation \eqref{eq:revision2}, 
one can show that
\begin{align}
\sum_{i=1}^n I(U_i; X_i)&\leq nf_1+nR+\sum_{i=1}^nI(\tilde{U}_i;\tilde{X}_i).
\label{eq:equiv-4}
\end{align}

Let $T$ be a random variable distributed uniformly over $[1:n]$ and independent of previously defined random variables. Then, inequalities~\eqref{eq:equiv-1}-\eqref{eq:equiv-4} can be equivalently written as 
\begin{align*}
I(X_T;Y_T|T)&\leq I(\tilde{X}_T;\tilde{Y}_T|T)+2f_1,\\
I(U_T;X_TY_T|T)&\leq I(\tilde{U}_T;\tilde{X}_T\tilde{Y}_T|T)+R+f_1,\\
I(U_T;X_T|T)&\leq I(\tilde{U}_T;\tilde{X}_T|T)+R+f_1,\\
I(U_T;Y_T|T)&\leq I(\tilde{U}_T;\tilde{Y}_T|T)+2R\color{black}+f_1.
\end{align*}
Let $f_2=I(T;X_T,Y_T)$. Observe that 
by Lemma~\ref{cuflem} at the end of the proof, $f_2$ vanishes as $\epsilon$ converges to zero..  Then the above set of equations imply that
\begin{align*}
I(X_T;Y_T)&\leq I(\tilde{X}_T;\tilde{Y}_T|T)+2f_1+f_2,\\
I(U_T;X_T,Y_T)&\leq I(\tilde{U}_T;\tilde{X}_T\tilde{Y}_T|T)+R+f_1+f_2,\\
I(U_T;X_T)&\leq I(\tilde{U}_T;\tilde{X}_T|T)+R+f_1+f_2,\\
I(U_T;Y_T)&\leq I(\tilde{U}_T;\tilde{Y}_T|T)+2R\color{black}+f_1+f_2.
\end{align*}
Let $X=X_T$ and $Y=Y_T$, and note that $X\rightarrow U_T\rightarrow Y$ forms a Markov chain. Then by the above inequalities for non-negative reals $\beta$, $\gamma$, and $\theta$ we have
\begin{align*}
I(X;Y)+&\min_{U: X-U-Y}\big[\beta I(U;XY)+\gamma I(U;X)+\theta I(U;Y)\big]\\
&\leq
I(\tilde{X}_T;\tilde{Y}_T|T)+\Big[\beta I(\tilde{U}_T;\tilde{X}_T\tilde{Y}_T|T)+\gamma I(\tilde{U}_T;\tilde{X}_T|T)+\theta I(\tilde{U}_T;\tilde{Y}_T|T)\Big]\\&\quad+(\beta+\gamma+2\theta\color{black})R+f_1+(f_1+f_2)(1+\beta+\gamma+\theta)
\\&\leq \max_{t}\left(
I(\tilde{X}_T;\tilde{Y}_T|T=t)+\Big[\beta I(\tilde{U}_T;\tilde{X}_T\tilde{Y}_T|T=t)+\gamma I(\tilde{U}_T;\tilde{X}_T|T=t)+\theta I(\tilde{U}_T;\tilde{Y}_T|T=t)\Big]\right)\\&\quad+(\beta+\gamma+2\theta\color{black})R+f_1+(f_1+f_2)(1+\beta+\gamma+\theta)
\\&= \max_{t}\left(
I(\tilde{X}_t;\tilde{Y}_t)+\Big[\beta I(\tilde{U}_t;\tilde{X}_t\tilde{Y}_t)+\gamma I(\tilde{U}_t;\tilde{X}_t)+\theta I(\tilde{U}_t;\tilde{Y}_t)\Big]\right)\\&\quad+(\beta+\gamma+2\theta\color{black})R+f_1+(f_1+f_2)(1+\beta+\gamma+\theta)
\\&= \max_{t}\left(
I(\tilde{X}_t;\tilde{Y}_t)+\min_{\tilde{U}: \tilde{X}_t-\tilde{U}-\tilde{Y}_t}\Big[\beta I(\tilde{U};\tilde{X}_t\tilde{Y}_t)+\gamma I(\tilde{U};\tilde{X}_t)+\theta I(\tilde{U};\tilde{Y}_t)\Big]\right)\\&\quad+(\beta+\gamma+2\theta\color{black})R+f_1+(f_1+f_2)(1+\beta+\gamma+\theta)
\\&\leq \max_{p(\tilde{x})}\bigg[I(\tilde{X};\tilde{Y})+\min_{\tilde{U}: \tilde{X}-\tilde{U}-\tilde{Y}}\big[\beta I(\tilde{U};\tilde{X}\tilde{Y})+\gamma I(\tilde{U};\tilde{X})+\theta I(\tilde{U};\tilde{Y})\big]\bigg]\\&\quad +(\beta+\gamma+2\theta\color{black})R+f_1+(f_1+f_2)(1+\beta+\gamma+\theta).
\end{align*}

Recall that both $f_1$ and $f_2$ converge to zero as $\epsilon$ goes to zero. Furthermore, by~\eqref{eq:conversethm-eps-2} the joint pmf of $(X_T, Y_T)$ converges to the desired pmf $p(x)p(y|x)$ as $\epsilon$ converges to zero. Therefore, to complete the proof, it remains to show that the expression
$$I(X;Y)+\min_{U: X-U-Y}\big[\beta I(U;XY)+\gamma I(U;X)+\theta I(U;Y)\big]$$
is a continuous function of the joint distribution on $(X,Y)$. Equivalently, we need to show that the function
\begin{align*}
g(\epsilon)=\min_{\underset{\underset{q(xy)\overset{\epsilon}{\approx} p(xy)}{X-U-Y}}{q_{uxy}}}\big[\beta I(U;XY)+\gamma I(U;X)+\theta I(U;Y)\big],
\end{align*}
for a fixed $p(xy)$, satisfies $\lim_{\epsilon\rightarrow0}g(\epsilon)=g(0)$. 

First, observe that $g(\epsilon)$ is a decreasing function of $\epsilon$; in particular, $g(0)\geq g(\epsilon)$ for all $\epsilon>0$. 
Thus, the limit $\lim_{\epsilon\longrightarrow0}g(\epsilon)$ exists and is at most $g(0)$. Second, for every $\epsilon>0$, the minimization over $U$, can be restricted to random variables $U$ with cardinality bound  $|\mathcal U|\leq |\mathcal{X}\times\mathcal Y|$. Let $p_\epsilon(x,y,u)$ be an optimal point in the minimization. Since $p_\epsilon(x,y,u)$  belongs to the  compact set of the probability simplex on a finite alphabet set, the set of optimal points has a limit point ${p}^*(xyu)$. We then have
$$g(0)\geq \lim_{\epsilon\rightarrow0}g(\epsilon)=\beta I_{{p}^*}(U;XY)+\gamma I_{{p}^*}(U;X)+\theta I_{{p}^*}(U;Y). $$
Moreover, we have $p^*(xy)=p(xy)$, and by the continuity of mutual information, $X-U-Y$ holds for the limit distribution $p^*(xyu)$ as well. Now by definition we have
\begin{align*}
g(0)&=\underset{U: X-U-Y}{\min}\big[\beta I(U;XY)+\gamma I(U;X)+\theta I(U;Y)\big]\\
&\leq\beta I_{{p}^*}(U;XY)+\gamma I_{{p}^*}(U;X)+\theta I_{{p}^*}(U;Y)\\
&=\lim_{\epsilon\rightarrow0}g(\epsilon)\\
& \leq g(0).
\end{align*}
This completes the proof.
\end{proof}

In the above proof we used the following lemma from~\cite{CuffsT}. 

\begin{lemma}[Entropy and timing information of nearly i.i.d.\ sequences~\cite{CuffsT}]\label{cuflem}
For any discrete random variables $W^n$ whose pmf satisfies
\begin{align*}
\Big\|p(w^n)-\prod_{t=1}^{n}\hat{p}_{t}(w_{t})\Big\|_{1}<\epsilon<\frac{1}{4},
\end{align*}
for some $\hat{p}_{1}(w), \dots, \hat{p}_{n}(w)$, we have
\begin{align*}
\sum_{t=1}^{n}I(W_t;W^{t-1})\leq 4n\epsilon(\log |{\cal{W}}| + \log\frac{1}{\epsilon}).
\end{align*}
Moreover, for any random variable $T\in\{1,\cdots,n\}$ independent of $W^n$,
\begin{align*}
I(W_T;T)\leq 4n\epsilon(\log |{\cal{W}}| +\log\frac{1}{\epsilon}).
\end{align*}
\end{lemma}

\subsubsection{Equivalent characterization for symmetric channels} \label{symmetric-case}
\begin{proof}[Proof of Theorem \ref{converseThm2}]
We claim that, for a symmetric channel, the maximum on the right hand side of~\eqref{eqn:tm2} is achieved at the uniform distribution $p^{\mathsf u}(\tilde{x})$. We prove this for binary input channels, and the proof for general channels is done in a similar way. More specifically, we show that if $p=p(\tilde{X}=0)$ and we let 
$$g(p)=I(\tilde{X};\tilde{Y})+\min_{\tilde{U}: \tilde{X}-\tilde{U}-\tilde{Y}}\big[\beta I(\tilde{U};\tilde{X}\tilde{Y})+\gamma I(\tilde{U};\tilde{X})+\theta I(\tilde{U};\tilde{Y})\big],$$
then $g(p)$ is maximized at $p=\frac 12$. To show this, we first claim that $g(p)=g(1-p)$. Take some $p(\tilde{x} \tilde{y}\tilde{u})$ such that $p=p(\tilde{X}=0)$ and $\tilde{X}-\tilde{U}-\tilde{Y}$. Let
$$\tilde{X}'=1-\tilde{X}, \tilde{Y}'=\pi_Y(Y), \tilde{U}'=\tilde{U},$$
where $\pi_Y$ is the permutation corresponding to permutation $\pi_X(0)=1, \pi_X(1)=0$ such that 
$$p_{Y|X}(\pi_Y(y)|\pi_X(x))=p_{Y|X}(y|x).$$ 
Clearly, all the mutual information terms remain the same for $\tilde{X}', \tilde{U}', \tilde{Y}'$, and by the symmetry of $p(\tilde y|\tilde x)$, the two channels $\tilde X\rightarrow \tilde Y$ and $\tilde X'\rightarrow \tilde Y'$ are the same. On the other hand, $p(\tilde{X}_1=0)=1-p$. This means that, for every choice of $\tilde U$ in the minimization of $g(p)$ there is a choice of $\tilde U$ in minimization of $g(1-p)$ that leads to the same answer. As a result, $g(1-p)= g(p)$. 

Let $p(\tilde x\tilde y\tilde u)$ be the distribution with $\tilde X-\tilde U-\tilde Y$, that achieves the minimum in $g(1/2)$, i.e.,
$$g(1/2) =I(\tilde{X};\tilde{Y})+\beta I(\tilde{U};\tilde{X}\tilde{Y})+\gamma I(\tilde{U};\tilde{X})+\theta I(\tilde{U};\tilde{Y}).$$ 
Now, fix the channel $p(\tilde y \tilde u|\tilde x)$, and instead of the uniform distribution on $\tilde X$ put the distribution $\Bern(p)$ on $\tilde X$. Denote the resulting distribution by $q_p(\tilde x \tilde y \tilde u)$. Then by definition we have
\begin{align}\label{eq:g-p-01}
g(p)\leq I_{q_p}(\tilde{X};\tilde{Y})+\beta I_{q_p}(\tilde{U};\tilde{X}\tilde{Y})+\gamma I_{q_p}(\tilde{U};\tilde{X})+\theta I_{q_p}(\tilde{U};\tilde{Y}).
\end{align}
We similarly have
\begin{align}\label{eq:g-p-02}
g(1-p)\leq I_{q_{(1-p)}}(\tilde{X};\tilde{Y})+\beta I_{q_{(1-p)}}(\tilde{U};\tilde{X}\tilde{Y})+\gamma I_{q_{(1-p)}}(\tilde{U};\tilde{X})+\theta I_{q_{(1-p)}}(\tilde{U};\tilde{Y}).
\end{align}

Observe that for a fixed $p(\tilde{u},\tilde{y}|\tilde{x})$ the expression
$$I(\tilde{X};\tilde{Y})+\beta I(\tilde{U};\tilde{X}\tilde{Y})+\gamma I(\tilde{U};\tilde{X})+\theta I(\tilde{U};\tilde{Y}),$$
is a concave function of $p(\tilde{x})$; this is because mutual information is concave in input distribution for a fixed channel implying that the first, third and fourth term are concave; the third term is equal to $I(\tilde{U};\tilde{X}\tilde{Y})=I(\tilde{U};\tilde{X})+I(\tilde{U};\tilde{Y}|\tilde{X})$ which is a concave term plus a linear term. Therefore, by~\eqref{eq:g-p-01} and~\eqref{eq:g-p-02} and this concavity we obtain 
$$g(p)=\frac{1}{2}(g(p)+g(1-p))\leq g(\frac{1}{2}).$$ 
This proves our claim.

Now by Theorem~\ref{converseThm} and the above claim, for any non-negative real numbers $\beta$, $\gamma$, and $\theta$ we have
\begin{align}I(X;Y)+&\min_{U: X-U-Y}\big[\beta I(U;XY)+\gamma I(U;X)+\theta I(U;Y)\big]\leq
\nonumber\\&
I(\tilde{X};\tilde{Y})+\min_{\tilde{U}: \tilde{X}-\tilde{U}-\tilde{Y}}\big[\beta I(\tilde{U};\tilde{X}\tilde{Y})+\gamma I(\tilde{U};\tilde{X})+\theta I(\tilde{U};\tilde{Y})\big],\label{refeq}
\end{align}
in which $p(\tilde{x})$ is fixed to be the uniform distribution. Since this inequality holds for all $\beta, \gamma$ and $\theta$, 
we find that $I(\tilde{X};\tilde{Y})\geq I(X;Y)$ and further
\begin{align*}\min_{U: X-U-Y}&\big[\beta I(U;XY)+\gamma I(U;X)+\theta I(U;Y)\big]\leq
\\&
\min_{\tilde{U}: \tilde{X}-\tilde{U}-\tilde{Y}}\big[\beta I(\tilde{U};\tilde{X}\tilde{Y})+\gamma I(\tilde{U};\tilde{X})+\theta I(\tilde{U};\tilde{Y})\big].
\end{align*}
Then by the definition of $\mathcal S(p( y| x), p(x))$,
the supporting hyperplane theorem would imply statement of the theorem, i.e.,
$$\mathcal{S}(p(\tilde y|\tilde x), p(\tilde x))\subseteq \mathcal{S}(p(y|x), p(x)),$$
if we show that $\mathcal{S}(p(y|x), p(x))$ is a convex set. 

Here we prove that $\mathcal{S}(p(y|x), p(x))$ is convex.  
Corresponding to any two points in $\mathcal{S}(p(y|x), p(x))$, one can find two random variables $U_1, U_2$ such that $X-U_1-Y$ and $X-U_2-Y$ form Markov chains. Let $T$ be a uniform random variable on $\{1,2\}$ and independent of $X, Y, U_1$, and $U_2$. \color{black} Let $U=(T, U_T)$. We clearly have $I(X;Y|U)=0$ and further
$I(U;XY)=\frac 12 (I(U_1;XY)+I(U_2;XY))$ etc. Therefore, we can use $U$ to show that the average of the two original points belongs to $\mathcal{S}(p(y|x), p(x))$. 
\end{proof}

\subsection{Point-to-point channel: BEC vs BSC}
\begin{proof}[Proof of Theorem \ref{thm3}]

(i)  In this part, we would like to compute $\mathcal{S}(\BEC(\epsilon), p^{\mathsf u})$ with uniform input distribution. \color{black} Take some $p(u|xy)$ such that $X-U-Y$, and assume without loss of generality that $p(u)>0$ for all $u\in \mathcal U$. Define $U'$ as a function of $U$ as follows:
\begin{align*}
U'=\begin{cases}
0\qquad \text{ if $U=u$ and }p(X=0|U=u)=1,\\
1\qquad \text{ if $U=u$ and  }p(X=1|U=u)=1,\\
e\qquad \text{ if $U=u$ and  }p(X=1|U=u)>0, P(X=0|U=u)>0.
\end{cases}
\end{align*}
Then we claim that $X-U'-Y$ forms a Markov chain. Observe that $I(X;Y|U'=0)=I(X;Y|U'=1)=0$, since $X$ is deterministic if $U'=0$ or $U'=1$. Moreover, $U'=e$ implies that $Y=e$ is deterministic and hence $I(X;Y|U'=e)=0$; this is because if for instance $p(Y=0|U'=e)>0$, then
$$p(Y=0|X=1)\geq p(U'=e|X=1)p(Y=0|U'=e)>0,$$
which is a contradiction. Therefore, $X-U'-Y$ forms a Markov chain.

Since $U'$ is a function of $U$ we have
\begin{align*}
I(U';XY)&\leq I(U;XY),\\ 
I(U';X)&\leq I(U;X),\\ 
I(U';Y)&\leq I(U;Y).
\end{align*} 
Therefore, in the definition of $\mathcal S(\BEC(\epsilon), p^{\mathsf u})$ without loss of generality we may assume that $U$ has the form of $U'$ defined above. This form is depicted in Fig.~\ref{fig2news}. Here $a, b, c, d\in [1-\epsilon,1]$ are arbitrary numbers with $\epsilon=1-ac=1-bd$. With the latter equations $U$ is indeed determined by the pair $(a, b)$ since $c$ and $d$ are computed in terms of $a, b$ and $\epsilon$.

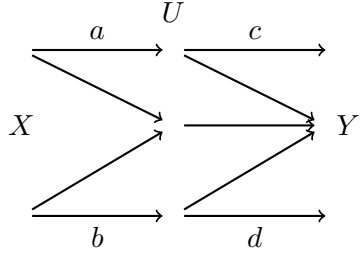
\begin{figure}
\begin{center}
\begin{tikzpicture}[scale=1., thick]
\node at (-2,1) (jointdist){$X$};
\node at (0,2.5){$U$};
\node at (0,2)  (commonrdn1){};
\node at (0,-0.2)  (commonrdn2){};
\node  at ( 0,1)  (ch) {};
\node at (2.3,1) (output){$Y$};
\node at ( 2,1) (dec) {};
\node at (-2,2) (enc1) {};
\node at (-2,-0.2) (enc2) {};
\draw [->] (enc1) --  (ch);
\draw [->] (enc2) -- (ch);
\draw [->] (ch) --(dec);
\draw [->]  (commonrdn1) -- node[above] {$c$}  (2,2) ;
\draw [->]  (enc1) -- node[above] {$a$}(commonrdn1);
\draw [->]  (commonrdn2) -- node[below] {$d$}  (2,-0.2) ;
\draw [->] (enc2) --node[below] {$b$} (commonrdn2);
\draw [->] (commonrdn1)--(dec);
\draw [->] (commonrdn2)--(dec);
\end{tikzpicture}
\end{center}
\caption{The form of $X\rightarrow U\rightarrow Y$ in the definition of $\mathcal{S}(\BEC(\epsilon), p^{\mathsf u})$.}
\label{fig2news}
\end{figure}

We claim that for the uniform input distribution $p(x)=p^{\mathsf u}(x)$ it suffices to consider the symmetric case with $a=b$ and  $c=d$. Observe that $p(x, y, u)$ is linear in terms of $a$ and $b$, e.g.,
\begin{align*}
p(X=0, Y=0, U=0) &=\frac{ac}{2}=\frac{1-\epsilon}{2},\\
P(X=0, Y=e, U=0) &= \frac{a(1-c)}{2} = \frac{a - ac}{2}=\frac{a-1+\epsilon}{2},\\
P(X=0, Y=e, U=e) &= \frac{1-a}{2}.
\end{align*}
On the other hand $I(U;XY)=H(XY) -H(XY|U)$ is a convex function when we linearly change the joint pmf of $p(x,u,y)$ while fixing $p(xy)$. Therefore, the value of $I(U; XY)$ at $(\frac{a+b}{2}, \frac{a+b}{2})$ is less than or equal to the average of its values  at $(a,b)$ and $(b,a)$. Moreover, by symmetry, the value of $I(U; XY)$ is the same at $(a,b)$ and $(b,a)$. We conclude that $I(U; XY)$ at $(\frac{a+b}{2}, \frac{a+b}{2})$ is not greater than this value at $(a,b)$. The same argument works for $I(U; X)$ and $I(U; Y)$ as well. 
Therefore, the three terms $I(U;XY), I(U;X)$ and $I(U;Y)$ are simultaneously minimized when $a=b$, and then $c=d$. 

Using the Markov chain condition $X- U- Y$ we have
\begin{align*}
I(U;XY)&=H(XY)-H(XY|U)\\
&=H(XY)-H(X|U)-H(Y|U)\\
&=H(XY)-H(X)-H(Y)+H(X)+H(Y)-H(X|U)-H(Y|U)\\
&=-I(X;Y)+I(X;U)+I(Y;U)
\end{align*}
Then, for the BEC channel with parameter $\epsilon$ we have 
$$I(U;XY)=-1+\epsilon +I(X;U)+I(Y;U).$$
Moreover, for $a=b\geq 1-\epsilon$ we have $H(X|U)=1-a$, and $H(Y|U)=ah(\frac{1-\epsilon}{a})$. Then
\begin{align*}
I(U;XY)&=h(\epsilon)+a-ah(\frac{1-\epsilon}{a})\\
I(U;X)&=a\\
I(U;Y)&=1-\epsilon+h(\epsilon)-ah(\frac{1-\epsilon}{a})
\end{align*}
The result then follows by a straightforward computation.

\vspace{.15in}
\noindent (ii)
We adapt the approach of Wyner to weighted sum calculations to prove our result.
Take the channel $\BSC(p)$ with uniform input distribution. Take some arbitrary auxiliary $U$ such that $X-U-Y$. We define two random variables as functions of $U$ by
$$A=p(X=0|U),\qquad B=p(Y=0|U).$$
Then we have
$$H(X|U)=\mathbb{E}[h(A)], \qquad H(Y|U)=\mathbb{E}[h(B)].$$
Furthermore,
$$p(X=0)=\mathbb{E}[A], \qquad p(Y=0)=\mathbb{E}[B].$$
Also 
$$p(X=0, Y=0|U)=p(X=0|U)p(Y=0|U)=AB,$$ and hence
$p(X=0, Y=0)=\mathbb{E}[AB].$
Therefore, we have
\begin{align}
I(U; XY)&= 1+h(p)-\mathbb{E}[h(A)]-\mathbb{E}[h(B)], \label{eq:u-xy-a-b-1}\\
I(U; X)&= 1-\mathbb{E}[h(A)],\\ 
I(U;Y)& =  1-\mathbb{E}[h(B)]. \label{eq:u-xy-a-b-3}
\end{align}
Here $A, B$ are real-valued random variables satisfying:
$$A, B\in [0,1],$$
$$\mathbb{E}[A]=\mathbb{E}[B]=\frac 12,$$
$$\mathbb{E}[AB]=\frac{1-p}{2}.$$


By the above equations to compute $\mathcal{S}(\BSC(\epsilon), p^{\mathsf u})$ we need to characterize the set 
$$\bigcup_{A, B}\big\{(b_1, b_2): b_1\leq \mathbb{E}[h(A)],\quad b_2\leq \mathbb{E}[h(B)]\big\},$$
where we take union over all  real-valued random variables $A, B$ satisfying the above constraints.
Equivalently, for any $\lambda\in [0,1]$ we need to compute
\begin{align}\label{eq:AB}
\max \lambda \E[h(A)] + \bar \lambda \E[h(B)],
\end{align}
over all $A, B$. We show that here the maximum occurs at two binary random variables $A$ and $B$ that correspond to $X\rightarrow U$ and $U\rightarrow Y$ being BSC channels. 

\begin{figure}[t]
\begin{center}
\includegraphics[scale=0.6]{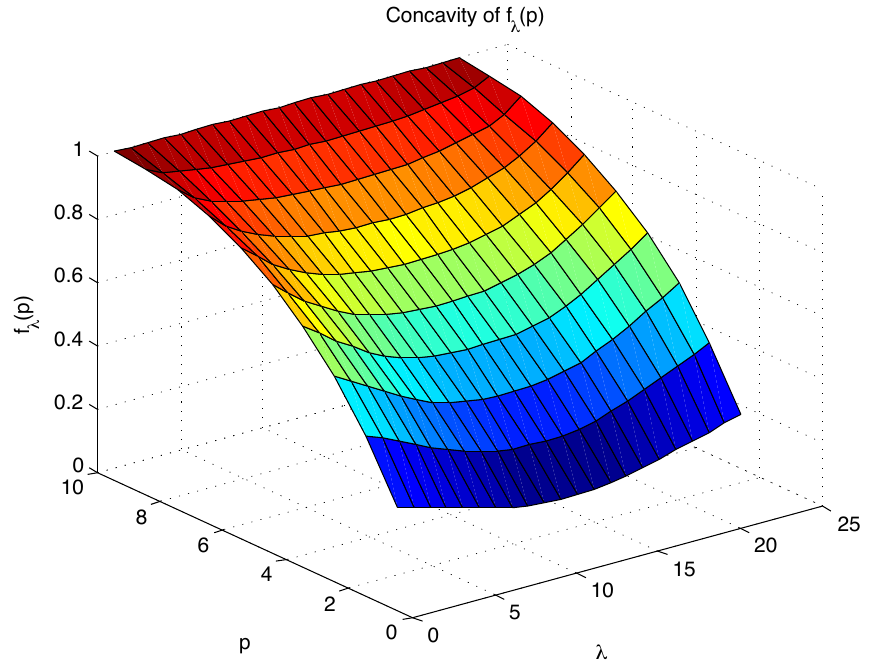}
\end{center}
\caption{Concavity of $f_\lambda(p)$ with respect to $0.05\leq p\leq 0.5$ and for all $0\leq\lambda\leq 1$\label{ConjPlot}.}
\end{figure}

Let $X\rightarrow U$ be a BSC with parameter $\alpha$, and let $U\rightarrow Y$ be another BSC with parameter $\beta$. We need the induced channel $X\rightarrow Y$ be a BSC with parameter $p\in [0,1/2]$. This is equivalent to
\begin{align}\label{eq:p}
p=\alpha*\beta=\alpha\bar \beta +\bar \alpha \beta.
\end{align}
For this special $U$ we get
\begin{align*}
\max_{\underset{\alpha*\beta=p}{\alpha,\beta}} \lambda H(X|U) +\bar \lambda H(Y|U)
  = &\max_{\underset{\alpha*\beta=p}{\alpha,\beta}} \lambda h(\alpha) + \bar \lambda h(\beta).
\end{align*}
We then make the following conjecture:
\begin{conjecture}\label{conjecture1}
Let $$f_\lambda(p)=\max_{\underset{\alpha*\beta=p}{\alpha, \beta}} \lambda h(\alpha) + \bar \lambda h(\beta).$$ Then $f_\lambda(p)$ is a concave function of $p$ for all $\lambda$, as plotted in Fig.~\ref{ConjPlot}. 
\end{conjecture}

Using this conjecture, we show that the answer to the maximization~\eqref{eq:AB} is also $f_\lambda(p)$ defined above. From the definitions it is clear that $f_\lambda(p)$ is a lower bound on~\eqref{eq:AB}. To prove inequality in the other direction, take $A,B$ with the above conditions. Assume that $(A, B) = (\alpha_i, \beta_j)$ happens with probability $q_{ij}$. We have
\begin{align*}
\lambda \E[h(A)] + \bar \lambda \E[h(B)] & = \sum_{i,j} q_{ij} [ \lambda h(\alpha_i) + \bar \lambda h(\beta_j)  ]\\
& \leq \sum_{i,j} q_{ij} f_\lambda(\alpha_i\bar\beta_j + \bar \alpha_i \beta_j)\\
& \leq f_\lambda\Big(   \sum_{i,j} q_{ij} (\alpha_i\bar\beta_j + \bar \alpha_i \beta_j)     \Big)\\
& = f_\lambda(\E[A] + \E[B] - 2\E[AB])\\
& = f_\lambda(p).
\end{align*}
Therefore, to compute~\eqref{eq:AB} we may restrict to auxiliary $U$ where $X\rightarrow U$ and $U\rightarrow Y$ are BSC channels with parameters $\alpha$ and $\beta$ respectively, with $p=\alpha* \beta$.  In this case, using equations~\eqref{eq:u-xy-a-b-1}-\eqref{eq:u-xy-a-b-3} we have
\begin{align*}
I(U;XY)&=1+h(p)-h(\alpha)-h(\beta),\\
I(U;X)&=1-h(\alpha),\\
I(U;Y)&=1-h(\beta).
\end{align*}
These give the desired result. 
\end{proof}

\subsection{Broadcast channel}
\begin{proof}[Proof of Theorem \ref{BCIn}]
The structure of the proof of this theorem is similar to that of Theorem~\ref{Pt2PtIn} and has three parts. 

\vspace{.15in}
\noindent
\textbf{Part (1)}: We define two protocols each of which induces a joint distribution on random variables that will be used in the proof.\\

\noindent
\emph{Protocol A [Not useful for coding].} Let $(W^{n},U^{n},V^{n},X^{n},\tilde{X}^{n},\tilde{Y}^n,\tilde{Z}^n,Y^{n},Z^{n})$ be $n$ i.i.d.\ repetitions of the joint pmf $p(w, u, v, x,\tilde{x}, \tilde{y}, \tilde{z}, y, z)$. Consider the following construction:
\begin{itemize}
\item To each sequence $w^{n}\in \mathcal W^n$ assign two random bin indices $g_{0}\in[1:2^{n\tilde{R}_{0}}]$ and $\omega\in[1:2^{nR}]$.

\item To each pair of sequences $(w^{n}, u^{n})$, assign a random bin index $g_{1}\in[1:2^{n\tilde{R}_{1}}]$.

\item To each pair of sequences $(w^{n}, v^{n})$, assign a random bin index $g_{2}\in[1:2^{n{\tilde{R}_{2}}}]$.

\item Consider two Slepian-Wolf decoders to estimate $(\hat{w}_1^{n}, \hat{u}^{n})$ and $(\hat{w}_2^{n}, \hat{v}^{n})$ from $(\omega, g_{1},\tilde{y}^{n})$ and $(\omega,g_{2},\tilde{z}^{n})$, respectively. Here we are considering $\tilde{y}^n$ and $\tilde{z}^n$ as side information, and $(\omega, g_1)$ and $(\omega, g_2)$ as random bins of the sources $({w}^{n}, {u}^{n})$ and $({w}^{n}, {v}^{n})$ that we want to decode. Note that $\hat{w}_1^{n}$ and $\hat{w}_2^{n}$ are reconstructions of $w^n$ by two different Slepian-Wolf decoders.
\end{itemize}
The constraints on the rates $R, \tilde R_0, \tilde R_1$ and $\tilde R_2$ for the success of the decoders will be imposed
later. The random pmf induced by the random binning, denoted by $Q$, can be expressed as follows:
\begin{align*}
Q(x^{n},y^{n},z^{n},w^n,&u^{n},v^{n},g_{[0:2]},\tilde{x}^n,\tilde{y}^n,\tilde{z}^n,\omega)\nonumber\\
&= p(x^{n})p(w^{n},u^{n},v^{n},\tilde{x}^n|x^{n})Q(g_{0},\omega |w^{n})Q(g_1|w^{n},u^{n})\nonumber\\
&\quad \times Q(g_2|w^{n},v^{n})p(\tilde{y}^n,\tilde{z}^n|\tilde{x}^n)Q^{\mathsf{sw}}(\hat{w}_1^{n},\hat{u}^{n}|g_{0}, g_1,\tilde{y}^n,\omega)\nonumber\\
&\quad \times Q^{\mathsf{sw}}(\hat{w}_2^{n},\hat{v}^{n}|g_{0},g_{2},\tilde{z}^n,\omega)p(y^{n}|w^n,u^{n},\tilde{y}^{n})p(z^{n}|w^n,v^{n},\tilde{z}^{n})\nonumber\\
&= p(x^{n})Q(w^{n},u^{n},v^{n},\tilde{x}^n,g_{[0:2]},\omega|x^{n})\nonumber\\
&\quad \times p(\tilde{y}^n,\tilde{z}^n|\tilde{x}^n)Q^{\mathsf{sw}}(\hat{w}_1^{n},\hat{u}^{n}|g_{0},g_1,\tilde{y}^n,\omega)\nonumber\\
&\quad \times Q^{\mathsf{sw}}(\hat{w}_2^{n},\hat{v}^{n}|g_{0},g_{2},\tilde{z}^n,\omega)p(y^{n}|w^n,u^{n},\tilde{y}^{n})p(z^{n}|w^n,v^{n},\tilde{z}^{n})\nonumber\\
&= p(x^{n})Q(g_{[0:2]},\omega |x^{n})Q(w^{n},u^{n},v^{n},\tilde{x}^n|x^n,g_{[0:2]},\omega)p(\tilde{y}^n,\tilde{z}^n|\tilde{x}^n)\nonumber\\
&\quad \times Q^{\mathsf{sw}}(\hat{w}_1^{n},\hat{u}^{n}|g_{0},g_1,\tilde{y}^n,\omega)Q^{\mathsf{sw}}(\hat{w}_2^{n},\hat{v}^{n}|g_{0},g_{2},\tilde{z}^n,\omega)\nonumber\\
&\quad \times p(y^{n}|w^n,u^{n},\tilde{y}^{n})p(z^{n}|w^n,v^{n},\tilde{
z}^{n}).\nonumber
\end{align*}

\vspace{.15in}
\noindent
\emph{Protocol B [Useful for coding after removing extra common randomnesses].} In this protocol we assume that the sender and receivers have access to the extra common randomness $(G_{0},G_{1},G_{2})$ where $G_{0},G_{1},G_{2}$ are mutually independent of $X^n$ and $\omega$. It is further assumed that $G_{0},G_{1}$ and $G_{2}$ are distributed uniformly over the sets $[1:2^{n\tilde{R}_{0}}]$, $[1:2^{n\tilde{R}_{1}}]$ and $[1:2^{n\tilde{R}_{2}}]$, respectively. Now we use the following protocol:
\begin{itemize}
 \item First, the sender having $(g_{[0:2]},\omega,x^{n})$ generates  $(w^{n},u^{n},v^{n},\tilde{x}^n)$ according to pmf $Q(w^{n},u^{n},v^{n},\tilde{x}^n|x^n,\\g_{[0:2]},\omega)$ of Protocol A, and sends $\tilde{x}^n$ over the memoryless broadcast channel $p(\tilde{y}^n, \tilde{z}^n|\tilde{x}^n)$. The first receiver gets $\tilde{y}^{n}$ and the second receiver gets $\tilde{z}^{n}$ from the channel. Having $(g_{0},g_1,\omega,\tilde{y}^{n})$, the first receiver uses the Slepian-Wolf decoder $Q^{\mathsf{sw}}(\hat{w}_1^{n}, \hat{u}^{n}|\omega,g_{0},g_1,\tilde{y}^{n})$ to estimate $(w^n, u^n)$. Similarly, the second receiver uses the Slepian-Wolf decoder $Q^{\mathsf{sw}}(\hat{w}_2^{n},\hat{v}^{n}|\omega,g_{0},g_{2},\tilde{z}^{n})$ to obtain an estimate of $(w^{n}, v^{n})$.
Here $\hat{w}_{1}^{n}$ and $\hat{w}_{2}^{n}$ are first and second receiver's estimate of $w^n$ respectively. 
 
 \item Having $(\tilde{y}^n,\hat{u}^n,\hat{w}_{1}^{n})$, the first receiver generates $y^{n}$ using $p(y^{n}|\tilde{y}^{n},\hat{u}^{n},\hat{w}_{1}^{n})=\prod_{i=1}^{n}p(y_i|\tilde{y}_i,\hat{u}_i,\hat{w}_{1i})$. Similarly the second receiver generates $z^{n}$  according to $ p(z^{n}|\tilde{z}^{n},\hat{v}^{n},\hat{w}_2^n)=\prod_{i=1}^{n}p(z_i|\tilde{z}_i,\hat{v}_i,\hat{w}_{2i})$.
\end{itemize}

The random pmf induced by the second protocol, denoted by $\hat{P}$, is equal to 
\begin{align*}
\hat{Q}(x^{n},y^{n},z^{n},w^n,u^{n},&v^{n},\hat{w}_{[1:2]}^n,\hat{u}^{n},\hat{v}^{n},g_{[0:2]},\tilde{x}^n,\tilde{y}^n,\tilde{z}^n,\omega)
\\&=p^{\mathsf u}(\omega)p^{\mathsf u}(g_{[0:2]})p(x^{n})Q(w^{n},u^{n},v^{n},\tilde{x}^n|x^n,g_{[0:2]},\omega)\\
& \quad \times p(\tilde{y}^n,\tilde{z}^n|\tilde{x}^n)
Q^{\mathsf{sw}}(\hat{w}_{1}^{n},\hat{u}^{n}|\omega,g_{[0:1]},\tilde{y}^{n})\\ 
& \quad \times Q^{\mathsf{sw}}(\hat{w}_{2}^{n},\hat{v}^{n}|\omega ,g_{0},g_2,\tilde{z}^{n})
 p(y^{n}|\tilde{y}^{n},\hat{u}^{n},\hat{w}_{1}^{n})p(z^{n}|\tilde{z}^{n},\hat{v}^{n},\hat{w}_{2}^{n}).
\end{align*}

\vspace{.15in}
\noindent
\textbf{Part (2)}: In this part we mention sufficient conditions under which the pmf's $Q$ and $\hat{Q}$ induced by the above protocols are approximately equal. The first step is to observe that $g_0, \omega, g_1$ and $g_2$ are bin indices of $w^n$, $w^n$, $w^nu^n$ and $w^nv^n$, respectively. Substituting $T=4$, $X_1\leftarrow W, X_2\leftarrow W, X_3\leftarrow WU$, $X_4 \leftarrow WV$ and $Z\leftarrow \emptyset$ in Theorem~1 of~\cite{OSRB}, we find that if
\begin{align}
R+\tilde{R}_{0}&<H(W|X),\nonumber\\
R+\tilde{R}_{0}+\tilde{R}_{1}&<H(WU|X),\nonumber\\
R+\tilde{R}_{0}+\tilde{R}_{2}&<H(WV|X),\nonumber\\
R+\tilde{R}_{0}+\tilde{R}_{1}+\tilde{R}_{2}&<H(WUV|X),\label{eq:bcc1}
\end{align}
then there exists $\epsilon^{(n)}_{0}\rightarrow 0$ such that $Q(g_{[0:2]},\omega|x^{n})\overset{\epsilon^{(n)}_{0}}{\approx}p^{\mathsf u}(\omega)p^{\mathsf u}(g_{[0:2]})=\hat{Q}(g_{[0:2]},\omega)$. This implies that 
\begin{align}
\hat{Q}(x^{n},w^n,u^n,v^n,\hat{w}^{n}_{1},\hat{u}^n,\hat{w}^{n}_{2},\hat{v}^n,g_{[0:2]},\tilde{x}^n,\tilde{y}^n,\tilde{z}^n,\omega)\overset{\epsilon^{(n)}_{0}}{\approx}Q(x^{n},w^{n},u^n,v^n,\hat{u}^n,\hat{w}^{n}_{1},\hat{v}^n,\hat{w}^{n}_{2},g_{[0:2]},\tilde{x}^n,\tilde{y}^n,\tilde{z}^n,\omega).\label{eq:bcc2}
\end{align}
Note that we have not yet included $y^n$ and $z^n$ in the above pmf's.

The next step is to find the conditions under which the Slepian-Wolf decoders of Protocol A work well with high probability. By the Slepian-Wolf theorem we need
\begin{align}
\tilde{R}_{0}+\tilde{R}_{1}+R&>H(WU|\tilde{Y}),\nonumber\\
\quad \tilde{R}_{1}&>H(U|W\tilde{Y}),\nonumber\\
\tilde{R}_{0}+\tilde{R}_{2}+R&>H(WV|\tilde{Z}),\nonumber\\
\quad \tilde{R}_{2}&>H(V|W\tilde{Z}).\label{eq:bcc3}
\end{align}
Then for an asymptotically vanishing sequence $\epsilon^{(n)}_{1}$, we have
\begin{align}
Q(x^{n},&w^n,u^{n},v^{n},g_{[0:2]},\tilde{x}^n,\tilde{y}^n,\tilde{z}^n,\omega ,\hat{w}^{n}_{1} ,\hat{u}^n ,\hat{w}^{n}_{2}, \hat{v}^n)\nonumber\\
&\quad \overset{\epsilon^{(n)}_{1}}{\approx}Q(x^{n},w^n,u^{n},v^{n},g_{[1:2]},\tilde{x}^n,\tilde{y}^n,\tilde{z}^n,\omega)\boldsymbol{1}\{\hat{w}^{n}_{1}=\hat w^n_2={w}^{n},\hat{u}^n=u^{n},\hat{v}^{n}=v^{n}\}\label{eq:bcc4}.
\end{align}
Using~\eqref{eq:bcc2} and~\eqref{eq:bcc4} and Lemma~3 of~\cite{OSRB} we have
\begin{align}
\hat{Q}(x^{n},&w^n,u^n,\hat{w}^{n}_{1},v^n,\hat{u}^n,\hat{w}^{n}_{2},\hat{v}^n,g_{[0:2]},\tilde{x}^n,\tilde{y}^n,\tilde{z}^n,\omega)\nonumber\\
&\quad\overset{\epsilon ^{(n)}_{0}+\epsilon^{(n)}_{1}}{\approx}Q(x^{n},u^{n},v^{n},w^n,g_{[0:2]},\tilde{x}^n,\tilde{y}^n,\tilde{z}^n,\omega)\boldsymbol{1}\{\hat{w}^{n}_{1}=\hat w^n_2=w^n,\hat{u}^n=u^{n},\hat{v}^n=v^{n}\}.
\end{align}
Moreover, the third part of Lemma~3 of~\cite{OSRB} implies that
\begin{align}
\hat{Q}(x^{n},w^n,&u^n,v^n,\hat{w}^{n}_{1},\hat{w}^{n}_{2},\hat{u}^n,\hat{v}^n,g_{[0:2]},\tilde{x}^n,\tilde{y}^n,\tilde{z}^n,\omega)p(z^n|\tilde{z}^n,\hat{v}^{n}, \hat{w}_2^n)p(y^n|\tilde{y}^n,\hat{u}^{n}, \hat{w}_1^n)\nonumber\\
&\overset{\epsilon ^{(n)}_{0}+\epsilon^{(n)}_{1}}{\approx}
Q(x^{n},u^{n},w^n,v^{n},g_{[0:2]},\tilde{x}^n,\tilde{y}^n,\tilde{z}^n,\omega)\boldsymbol{1}\{\hat{w}^{n}_{1}=\hat{w}^{n}_{2}=w^n,\hat{u}^n=u^{n},\hat{v}^n=v^{n}\}\nonumber\\
&\qquad\qquad \times p(z^n|\tilde{z}^n,\hat{w}^{n}_{2},\hat{v}^{n})p(y^n|\tilde{y}^n,\hat{w}^{n}_{1},\hat{u}^{n})\nonumber\\
&\quad =Q(x^{n},u^{n},w^n,v^{n},g_{[0:2]},\tilde{x}^n,\tilde{y}^n,\tilde{z}^n,\omega)\boldsymbol{1}\{\hat{w}^{n}_{1}=\hat{w}^{n}_{2}=w^n,\hat{u}^n=u^{n},\hat{v}^n=v^{n}\}\nonumber\\
&\qquad \qquad \times p(z^n|\tilde{z}^n{w}^{n}{v}^{n})p(y^n|\tilde{y}^n{w}^{n}{u}^{n}).
\end{align}
Therefore, 
\begin{align}
\hat{Q}&(x^{n},y^{n},z^n,w^n,u^n,\hat{w}^{n}_{1},v^n,\hat{w}^{n}_{2},\hat{u}^n,\hat{v}^n,g_{[0:2]},\tilde{x}^n,\tilde{y}^n,\tilde{z}^n,\omega)\overset{\epsilon ^{(n)}_{0}+\epsilon^{(n)}_{1}}{\approx}\nonumber\\
&\qquad Q(x^{n},y^{n},z^n,w^n,u^{n},v^{n},g_{[0:2]},\tilde{x}^n,\tilde{y}^n,\tilde{z}^n,\omega)\boldsymbol{1}\{\hat{w}^{n}_{1}=\hat{w}^{n}_{2}=w^n,\hat{u}^n=u^{n},\hat{v}^n=v^{n}\}.
\end{align}
Finally, using the second item in part~1 of Lemma~3 of~\cite{OSRB} we conclude that
\begin{align}
\hat{Q}(g_{[0:2]},x^n,y^n,z^n)\overset{\epsilon ^{(n)}_{0}+\epsilon^{(n)}_{1}}{\approx}Q(g_{[0:2]},x^n,y^n,z^n).
\end{align}
In particular, the marginal pmf of $(X^n,Y^n,Z^n)$ of the right hand side of this expression is equal to $p(x^n,y^n,z^n)$, which is the desired distribution.

\vspace{.15in}
\noindent
\textbf{Part(3)}:
In the above protocol we assumed that the sender and receivers have access to external randomnesses $G_{[0:2]}$
which are not present in the model. To eliminate these extra common randomnesses, we will fix particular instances $g_{[0:2]}$ of $G_{[0:2]}$ and show that the same protocol works even if we fix $G_{[0:2]}=g_{[0:2]}$. To prove this note that by letting $G_{[0:2]}=g_{[0:2]}$, the induced pmf $\hat{Q}(x^n,y^n,z^n)$ changes to the conditional pmf $\hat{Q}(x^n,y^n,z^n|g_{[0:2]})$. But if $G_{[0:2]}$ is almost independent of $(X^n,Y^n,Z^n)$, the conditional pmf $\hat{Q}(x^n,y^n,z^n|g_{[0:2]})$ would be close to the desired distribution as well. To
obtain the independence, we again use Theorem~1 of~\cite{OSRB}. Substituting $T = 3$, $X_1\leftarrow W$, $X_2\leftarrow WU$, $X_3 \leftarrow WV$, and
$Z \leftarrow  XYZ$ in Theorem~1 of~\cite{OSRB}, we find that if
\begin{align}
\tilde{R}_{0}&<H(W|XYZ),\nonumber\\
\tilde{R}_{0}+\tilde{R}_{1}&<H(WU|XYZ),\nonumber\\
\tilde{R}_{0}+\tilde{R}_{2}&<H(WV|XYZ),\nonumber\\
\tilde{R}_{0}+\tilde{R}_{1}+\tilde{R}_{2}&<H(WUV|XYZ),\label{eq:bcc5}
\end{align}
then $Q(x^n,y^n,z^n,g_{[0:2]})\overset{\epsilon^{(n)}_{2}}{\approx}p^{\mathsf u}(g_{[0:2]})p(x^n,y^n,z^n)$, for some asymptotically vanishing $\epsilon ^{(n)}_{2}$. Thus, by triangular inequality for total variation, we have $\hat{Q}(x^n,y^n,z^n,g_{[0:2]})\overset{\epsilon ^{(n)}}{\approx}p^{U}(g_{[0:2]})p(x^n,y^n,z^n)$, where $\epsilon^{(n)}=\sum_{i=0}^{2}\epsilon^{(n)}_{i}$. From the definition of total variation distance random pmf's, the average of the total variation distance between $\hat{q}(x^n,y^n,z^n,g_{[0:2]})$ and $p^{\mathsf u}(g_{[0:2]})p(x^n,y^n,z^n)$ over all random binning is small. Thus, there exists a fixed binning with the corresponding pmf $\hat{q}$ such that $\hat{q}(x^n,y^n,z^n,g_{[0:2]})\overset{\epsilon ^{(n)}}{\approx}p^{\mathsf u}(g_{[0:2]})p(x^n,y^n,z^n)$. Next, Lemma~3 of~\cite{OSRB} guarantees the existance of an instance $g_{[0:2]}$ such that $$\hat{q}(x^n,y^n,z^n|g_{[0:2]})\overset{2\epsilon ^{(n)}}{\approx}p(x^n,y^n,z^n).$$ 
Then the extra shared randomness $G_{[0:2]}$ can be eliminated by fixing it to be $G_{[0:2]}=g_{[0:2]}$.

Finally, observe that the rate region in the theorem is seen to be equivalent to that given by equations~\eqref{eq:bcc1}, \eqref{eq:bcc3} and~\eqref{eq:bcc5} after eliminating $\tilde{R}_0,\tilde{R}_1,\tilde{R}_2$ using Fourier-Motzkin elimination. 


\end{proof}

\subsection{An infeasibility result for exact channel simulation}
\begin{proof}[Proof of Theorem \ref{exact-converse}]
Let $\Phi$ be a function that takes in an arbitrary discrete channel and returns a real number. Assume that $\Phi(p(y|x))$ that satisfies additivity and data processing properties, namely,
$$\Phi\big(\prod_{i=1}^np(y_i|x_i)\big)=n \Phi\big(p(y|x)\big)$$
and 
$$\Phi(p(y|x))\geq \Phi(p(b|a)), ~~~ \text{if } ~~~p(b|a)=\sum_{x,y}p(b|y)p(y|x)p(x|a) \text{ for some } p(x|a) \emph{ and } p(b|y).$$
We claim that if $p(y|x)$ can be exactly simulated from $p(\tilde{y}|\tilde{x})$ with no shared randomness, we have $\Phi(p(y|x))\geq \Phi(p(\tilde y|\tilde x))$. To show this assume that 
 there is $n$ and encoding and decoding maps which result in a joint pmf $q(x^n, y^n, \tilde x^n, \tilde y^n)$ such that~\eqref{eq:n-sim-exact-4} holds. 
We then have
\begin{align}n\Phi(p(y|x))&=\Phi(p(y^n|x^n))\label{eqn:alpha1}\\&
\geq \Phi(p(\tilde y^n|\tilde x^n))\label{eqn:alpha2}\\&=n\Phi(p(\tilde y|\tilde x))\label{eqn:alpha3},\end{align}
where~\eqref{eqn:alpha1} and~\eqref{eqn:alpha3} follow from the additivity of $\Phi$ for product channels, and~\eqref{eqn:alpha2} follows from the data processing property of $\Phi$. 

Next, assume that  $\Phi(p(y|x))$ is also quasi-convex in $p(y|x)$. We claim that if $p(y|x)$ can be exactly simulated from $p(\tilde{y}|\tilde{x})$ with infinite shared randomness, we have $\Phi(p(y|x))\geq \Phi(p(\tilde y|\tilde x))$. Assume that there is $n$ and encoding and decoding maps which result in a joint pmf $q(x^n, y^n, \tilde x^n, \tilde y^n, \omega)$ such that~\eqref{eq:n-sim-exact-4} holds. Since $p(y^n|x^n)=\sum_\omega p(y^n|x^n,\omega)$, by quasi-convexity of $\Phi$, 
there is some choice for $\omega=\omega^*$ such that $$\Phi(p(y^n|x^n))\geq \Phi(p(y^n|x^n,\omega^*)).$$
Then, using the fact that $p(\tilde y^n|\tilde x^n,\omega^*)=p(\tilde y^n|\tilde x^n)$,
we can similarly write
\begin{align}n\Phi(p(y|x))=\Phi(p(y^n|x^n))\geq \Phi(p(y^n|x^n,\omega^*))
\geq \Phi(p(\tilde y^n|\tilde x^n))=n\Phi(p(\tilde y|\tilde x)).\end{align}
\color{black}

Note that $\Phi(p(y|x))=C_\alpha(p(y|x))$ satisfies the additivity by Theorem~\ref{thm:prodal}, data processing  by Theorem~\ref{thm:dataal} and quasi-convexity by Lemma~\ref{lemma:rev-added}.  This concludes the proof for capacity of order $\alpha$. 

It remains to show that $\Phi(p(y|x))=\mathsf{Diam}_\alpha(p(y|x))$ satisfies the additivity, data processing and quasi-convexity properties:

\noindent\emph{Data processing:} If $p(z|x)=\sum_{y}p(y|x)p(z|y)$, then by the data processing property of $\alpha$-R\'enyi divergence we have
$$\mathsf{Diam}_\alpha(p(y|x))=\max_{p(x), q(x)}D_{\alpha}(p(y)\|q(y)) \geq \max_{p(x), q(x)}D_{\alpha}(p(z)\|q(z))=\mathsf{Diam}_\alpha(p(z|x)).$$
If $p(y|w)=\sum_x p(x|w)p(y|x)$, then any $p(w)$ and $q(w)$ correspond to some $p(x)$ and $q(x)$. Therefore, 
$$\mathsf{Diam}_\alpha(p(y|w))=\max_{p(w), q(w)}D_{\alpha}(p(y)\|q(y)) \leq \max_{p(x), q(x)}D_{\alpha}(p(y)\|q(y))=\mathsf{Diam}_\alpha(p(y|x)).$$

\noindent
\emph{Additivity:} First, observe that by the quasi-convexity of $D_{\alpha}$ in its arguments (see e.g., \cite{Ervin}), we have
\begin{align}\mathsf{Diam}_\alpha(p(y|x))&=\max_{p(x), q(x)}D_{\alpha}(p(y)\|q(y))\nonumber
\\&=\max_{x_1,x_2}D_{\alpha}\big(p(y|x_1)\big\|p(y|x_2)\big)\label{eqn:diam2}.\end{align}
Therefore, 
\begin{align}
\mathsf{Diam}_\alpha(p(y^n|x^n))&=\max_{x_1^n,x_2^{n}}D_{\alpha}\big(p(y^n|x_1^n)\big\|p(y^n|x_2^{n})\big)\nonumber
\\&=\max_{x_1^n,x_2^{n}}\sum_{i=1}^n D_{\alpha}\big(p(y_i|x_{1i})\big\|p(y_i|x_{2i})\big)\label{eqn:dexpl}
\\&=n\cdot\max_{x_1,x_2}D_{\alpha}\big(p(y|x_1)\big\|p(y|x_2)\big)\nonumber
\\&=n\cdot\mathsf{Diam}_\alpha(p(y|x)),\nonumber\end{align}
where in~\eqref{eqn:dexpl} we use the fact that $p(y^n|x_1^n)=\prod_{i=1}^np(y_i|x_{1i})$ and similarly for $p(y^n|x_2^n)$.

\noindent
\emph{Quasi-convexity:}  This follows from  the quasi-convexity of $D_{\alpha}$ in its arguments (see e.g., \cite{Ervin}). If $p(y|x)=\sum_wp(y|x,w)p(w)$, then for every $x_1, x_2$, we have
$$D_{\alpha}\big(p(y|x_1)\big\|p(y|x_2)\big)\leq \max_w D_{\alpha}\big(p(y|x_1,w)\big\|p(y|x_2,w)\big).$$
Hence, characterization of channel diameter in \eqref{eqn:diam2}, we have
$$\max_{x_1,x_2}D_{\alpha}\big(p(y|x_1)\big\|p(y|x_2)\big)\leq 
\max_w\max_{x_1,x_2}D_{\alpha}\big(p(y|x_1,w)\big\|p(y|x_2,w)\big).$$\color{black}
\end{proof}

\subsection{Exact simulation of a BSC channel from a BEC channel}
\begin{proof}[Proof of Theorem \ref{BSCCAP}]
By Theorem~\ref{exact-converse}, $\BSC(p)$ can be exactly simulated from $\BEC(\epsilon)$ with infinite shared randomness only if
$$C_\infty(\BEC(\epsilon))\geq C_{\infty}(\BSC(p)).$$
Using equation~\eqref{RenyMut}, it is easy to verify that $C_\infty(\BEC(\epsilon))=\log(2-\epsilon)$ and 
\begin{align*}
C_{\infty}(\BSC(p))=\log(2\max\{p,\bar{p}\}).
\end{align*}
Thus, we should have $2\max(p,\bar{p})\leq 2-\epsilon$. Since $p\in[0,1/2]$, we get $2(1-p)\leq 2-\epsilon$, or 
$p\geq\epsilon/2$. On the other hand, a degradation strategy shows that any $p\geq\epsilon/2$ is achievable (without any need for shared randomness). This completes the proof.
\end{proof}
\subsection{Exact simulation of a BIBO channel from a BIBO channel}

\begin{figure}
\begin{center}
\includegraphics[scale=0.5]{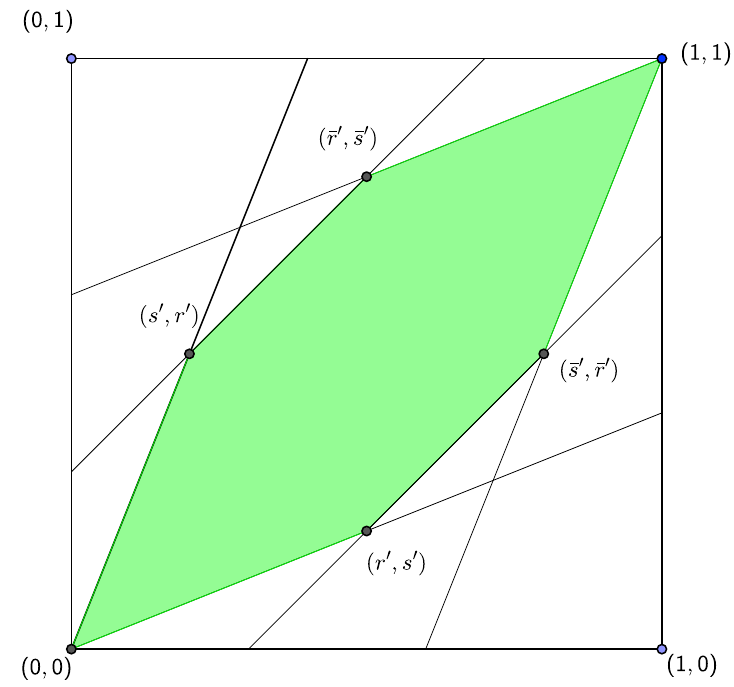}
\caption{The boundary of the simulation region for the exact channel simulation of BIBO  channel from another BIBO channel.}
\label{Capregion}
\end{center}
\end{figure}

\begin{proof}[Proof of Theorem \ref{BSCCAP2}]

\emph{Acheivability:} We first show that  any point  $(r,s)$ inside any of the two parallelograms is achievable. The parallelogram with vertices $\{(r,s), (\bar{r},\bar{s}), (0,0), (1,1)\}$ is achievable as follows: fix $\tilde{X}=X$. Then there are decoder strategies for achieving any of these four vertices of the parallelogram if we use $Y=\tilde{Y}, Y=1-\tilde{Y}, Y=0$ and $Y=1$. The whole parallelogram is achievable by time-sharing between these vertices using \emph{private} randomness at the decoder. The parallelogram with vertices $\{(s,r), (\bar{s},\bar{r}), (0,0), (1,1)\}$ is achievable in a similar way if we fix $\tilde{X}=1-X$ instead.

Thus, if shared randomness is not available, the union of the two parallelograms is achievable. If shared randomness is available, the convex hull of the region, which is the convex polygon is achievable.\color{black}\\

\noindent
\emph{Converse:} It suffices to prove the converse for the case of infinite shared randomness. It is clear that the simulation region when there is no shared randomness cannot exceed that when shared randomness exists. \color{black}

By~\eqref{RenyMut} for a BIBO channel $p(y|x)$ with parameters $(r, s)$ we have
 \begin{align*}
 C_{\infty}(p(y|x))&=\log(\max\{r,s\}+\max\{\bar{r}, \bar{s}\})
\\ &= \log\left(\max\{r+\bar{s},s+\bar{r}\}\right).
 \end{align*}
Thus, by Theorem~\ref{exact-converse}, the possibility of simulation gives 
$$\max\{r+\bar{s},s+\bar{r}\}\leq \max\{r'+\bar{s}',s'+\bar{r}'\},$$
or
$\max\{r-s,s-r\}\leq \max\{r'-s',s'-r'\}$. Equivalently, we have 
\begin{align}
|r-s|\leq |r'-s'|.
\label{Outin1}
\end{align}

Similarly, using~\eqref{eqn:diam2}, for such a channel $p(y|x)$ we have
 \begin{align*}\mathsf{Diam}_\infty(p(y|x))&=\max_{x_1, x_2}D_{\infty}(p(y|x_1)\|q(y|x_2))
\\&=\max\big\{D_\infty((r,\bar{r})\|(s,\bar{s})) , D_\infty((s,\bar{s})\|(r,\bar{r})) \big\}
\\&=\log\max\Big\{\frac{r}{s}, \frac{\bar{r}}{\bar{s}}, \frac{s}{r}, \frac{\bar{s}}{\bar{r}}\Big\}.
 \end{align*}
Therefore, again by Theorem~\ref{exact-converse}, the possibility of channel simulation gives 
\begin{align}
     \max\Big\{\frac{r}{s},\frac{s}{r},\frac{\bar{r}}{\bar{s}},\frac{\bar{s}}{\bar{r}} \Big\}&\leq \max\Big\{\frac{r'}{s'},\frac{s'}{r'},\frac{\bar{r'}}{\bar{s'}},\frac{\bar{s'}}{\bar{r'}} \Big\}.\label{Outin2}
     \end{align}
Equations (\ref{Outin1}) and (\ref{Outin2}) imply that $(r,s)$ is in the area depicted in Fig.~\ref{Capregion} for the given pair $(r',s')$. In particular, equation (\ref{Outin1}) gives the two edges that are parallel to the line $r=s$, and equation  (\ref{Outin2}) gives the four side boundaries of the region (see Fig.~\ref{Capregion}). This completes the proof.

\end{proof}

\section*{Acknowledgement}
The authors would like to thank the anonymous reviewer for valuable comments and suggestions to improve the quality of the paper.

\end{document}